\documentclass[11pt,a4paper]{article}
\usepackage[lmargin=1.0in,rmargin=1.0in,bottom=1.0in,top=1.0in,twoside=False]{geometry}
\usepackage[english]{babel}
\usepackage[shortlabels,]{enumitem}
\usepackage{utf8math}

\usepackage{fontaxes}
\usepackage{trimspaces}
\usepackage{nccfoots}
\usepackage{setspace}
\usepackage{inconsolata}
\usepackage{libertine}
\usepackage{amsmath,amssymb,amsthm}
\usepackage{mathtools}
\usepackage{comment}
\usepackage[all,defaultlines=4]{nowidow}
\usepackage{microtype}
\usepackage{mathrsfs}
\usepackage{mathtools}
\usepackage{marvosym}
\usepackage{thm-restate}
\usepackage{tikz}
\usepackage{todonotes}
\usetikzlibrary{calc}
\usetikzlibrary{shapes}
\definecolor{darkgreen}{rgb}{10,117,28}
\usepackage{wrapfig}
\usepackage{diagbox}
\usepackage{stackengine}

\definecolor{blue}{rgb}{0.1,0.2,0.5}
\definecolor{brown}{rgb}{0.6,0.6,0.2}
\usepackage[ocgcolorlinks, linkcolor={blue}, citecolor={blue}]{hyperref}
\usepackage{cleveref}

\crefformat{page}{#2page~#1#3}%
\Crefformat{page}{#2Page~#1#3}%
\crefformat{equation}{#2(#1)#3}%
\Crefformat{equation}{#2(#1)#3}%
\crefformat{figure}{#2Figure~#1#3}%
\Crefformat{figure}{#2Figure~#1#3}%
\crefformat{section}{#2Section~#1#3}
\Crefformat{section}{#2Section~#1#3}
\crefformat{chapter}{#2Chapter~#1#3}
\Crefformat{chapter}{#2Chapter~#1#3}
\crefformat{chapter*}{#2Chapter~#1#3}
\Crefformat{chapter*}{#2Chapter~#1#3}
\crefformat{part}{#2Part~#1#3}
\Crefformat{part}{#2Part~#1#3}
\crefformat{enumi}{#2(#1)#3}
\Crefformat{enumi}{#2(#1)#3}

\newtheorem{theorem}{Theorem}[section]
\crefformat{theorem}{#2Theorem~#1#3}
\Crefformat{theorem}{#2Theorem~#1#3}
\newtheorem{lemma}[theorem]{Lemma}
\crefformat{lemma}{#2Lemma~#1#3}
\Crefformat{lemma}{#2Lemma~#1#3}

\crefformat{conjecture}{#2Conjecture~#1#3}
\Crefformat{conjecture}{#2Conjecture~#1#3}

\newcommand{\newtheoremwithcrefformat}[2]{%
  \newtheorem{#1}{#2}[section]%
  \crefformat{#1}{##2\MakeUppercase#1~##1##3}%
  \Crefformat{#1}{##2\MakeUppercase#1~##1##3}%
}

\newtheoremwithcrefformat{proposition}{Proposition}
\newtheoremwithcrefformat{observation}{Observation}
\newtheoremwithcrefformat{corollary}{Corollary}
\newtheoremwithcrefformat{claim}{Claim}
\newtheoremwithcrefformat{example}{Example}
\newtheoremwithcrefformat{remark}{Remark}
\newtheoremwithcrefformat{definition}{Definition}

\makeatletter
\def\ifenv#1{
   \def\@tempa{#1}%
   \ifx\@tempa\@currenvir
      \expandafter\@firstoftwo
    \else
      \expandafter\@secondoftwo
   \fi
}
\let\wfs@comment@comment\comment
\let\comment\@undefined
\newcommand{\untoto}{\let\toto\@undefined}
\usepackage{changes}
\let\wfs@changes@comment\comment
\let\comment\@undefined

\newcommand\comment{%
    \ifthenelse{\equal{\@currenvir}{comment}}
    {\wfs@comment@comment}
    {\wfs@changes@comment}%
}
\makeatother

\usepackage{anyfontsize}

\renewcommand{\top}{\textsf{top}}

\renewcommand{\phi}{\varphi}

\renewcommand{\leq}{\leqslant}
\renewcommand{\geq}{\geqslant}
\renewcommand{\le}{\leqslant}
\renewcommand{\ge}{\geqslant}
\newcommand{\tup}[1]{\bar{#1}}

\newcommand{\trans}{\intercal}
\newcommand{\ilpform}{$\spadesuit$}
\newcommand{\ineqform}{$\clubsuit$}
\newcommand{\stochform}{$\diamondsuit$}
\newcommand{\sol}{\mathsf{Sol}}
\newcommand{\opt}{\mathsf{opt}}
\newcommand{\Z}{\mathbb{Z}}
\newcommand{\R}{\mathbb{R}}
\newcommand{\N}{\mathbb{N}}
\newcommand{\Q}{\mathbb{Q}}
\newcommand{\proximity}{\mathsf{proximity}_\infty}

\usepackage{stmaryrd}

\newcommand{\frc}{\star}
\newcommand{\itg}{\diamond}
\newcommand{\itgo}{\oblong}

\newcommand{\cleq}{\sqsubseteq}

\newcommand{\Graver}{\mathcal{G}}
\newcommand{\depth}{\mathsf{depth}}
\newcommand{\tdP}{\mathsf{td}_{\mathsf{P}}}
\newcommand{\GP}{G_{\mathsf{P}}}

\renewcommand{\setminus}{-}

\DeclareMathOperator{\intcone}{int.cone}
\DeclareMathOperator{\cone}{cone}
\DeclareMathOperator{\lcm}{lcm}
\DeclareMathOperator{\conv}{conv}

\newcommand{\Oh}{\mathcal{O}}

\usepackage{accents}
\newcommand{\dbtilde}[1]{\accentset{\approx}{#1}}
\newcommand{\wo}[1]{\overline{#1}}
\newcommand{\wh}[1]{\widehat{#1}}
\newcommand{\wt}[1]{\widetilde{#1}}

\newcommand{\ERCagreement}{
\vspace{-6pt}

\noindent
{\begin{minipage}[t]{.75\textwidth}\small This paper is a part of project {\sc{Total}} that has received funding from the European Research Council (ERC) under the European Union's Horizon 2020 research and innovation programme (grant agreement No 677651). \end{minipage}\hfill\begin{minipage}[t]{.22\textwidth}\raisebox{-42pt}{\includegraphics[width=\textwidth]{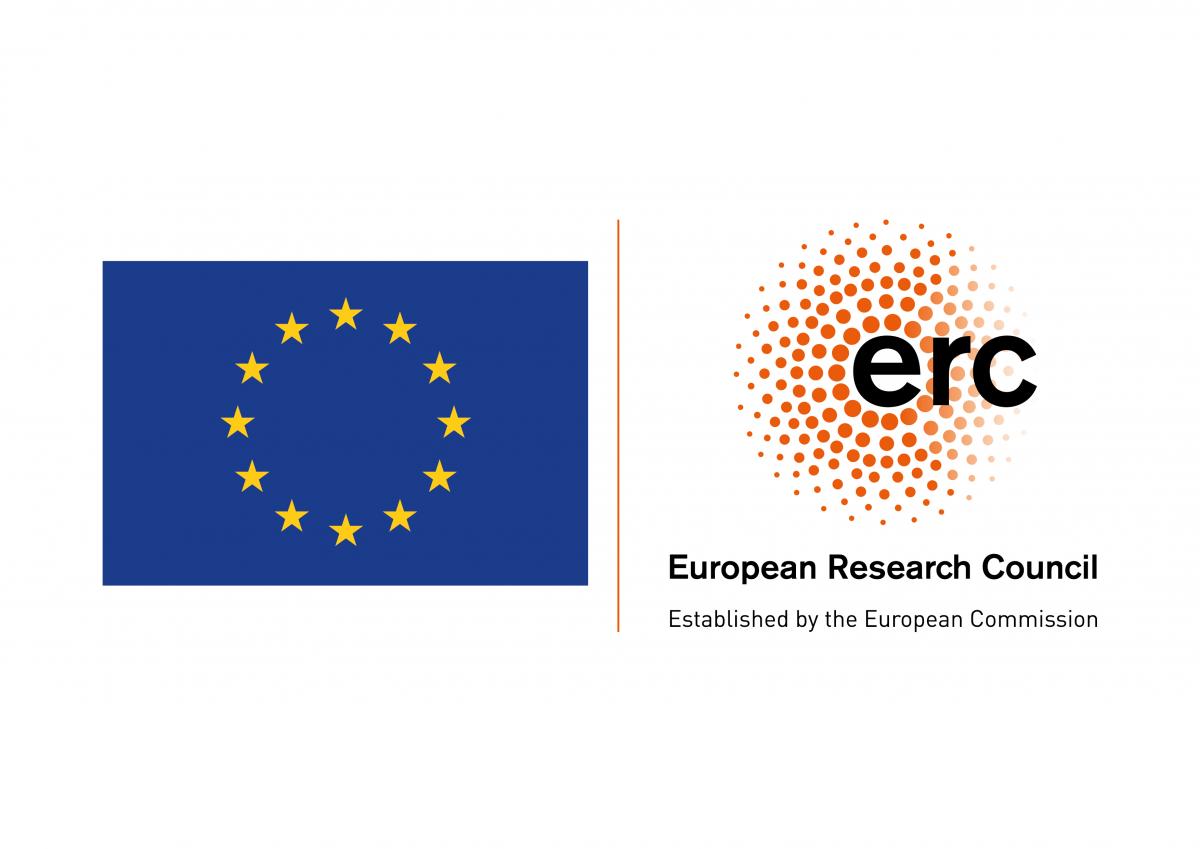}}\\\end{minipage}}}

\makeatletter
\newcommand{\@abbrev}[3]{
  \def\c@a@def##1{
      \if ##1.
        \relax
      \else
        \@ifdefinable{\@nameuse{#1##1}}{\@namedef{#1##1}{#2##1}}
        \expandafter\c@a@def
      \fi
    }
  \c@a@def #3.
}
\@abbrev{bb}{\mathbb}{ABCDEFGHIJKLMNOPQRSTUVWXYZ}
\@abbrev{bf}{\mathbf}{ABCDEFGHIJKLMNOPQRSTUVWXYZabcdefghijklmnopqrstuvwxyz}
\@abbrev{bit}{\boldsymbol}{ABCDEFGHIJKLMNOPQRSTUVWXYZabcdefghijklmnopqrstuvwxyz}
\@abbrev{cal}{\mathcal}{ABCDEFGHIJKLMNOPQRSTUVWXYZ}
\@abbrev{frak}{\mathfrak}{ABCDEFGHIJKLMNOPQRSTUVWXYZabcdefghijklmnopqrstuvwxyz}
\@abbrev{rm}{\mathrm}{ABCDEFGHIJKLMNOPQRSTUVWXYZabcdefghijklmnopqrstuvwxyz}
\@abbrev{scr}{\mathscr}{ABCDEFGHIJKLMNOPQRSTUVWXYZ}
\@abbrev{sf}{\mathsf}{ABCDEFGHIJKLMNOPQRSTUVWXYZabcdefghijklmnopqrstuvwxyz}
\@abbrev{Alg}{\Alg}{ABCDEFGHIJKLMNOPQRSTUVWXYZ}
\@abbrev{Str}{\Str}{ABCDEFGHIJKLMNOPQRSTUVWXYZ}
\@abbrev{set}{\mathbb}{ABCDEFGHIJKLMNOPQRSTUVWXYZ}
\@abbrev{tup}{\tup}{ABCDEFGHIJKLMNOPQRSTUVWXYZabcdefghijklmnopqrstuvwxyz}
\@abbrev{tld}{\tilde}{ABCDEFGHIJKLMNOPQRSTUVWXYZabcdefghijklmnopqrstuvwxyz}
\makeatother

\usepackage{wasysym}
\usepackage{xspace}

\title{Efficient sequential and parallel algorithms for multistage stochastic integer  programming using proximity}

\author{
Jana Cslovjecsek\thanks{EPFL, Switzerland, \texttt{jana.cslovjecsek@epfl.ch}}
\and
Friedrich Eisenbrand\thanks{EPFL, Switzerland, \texttt{friedrich.eisenbrand@epfl.ch}}
\and
Micha\l{} Pilipczuk\thanks{University of Warsaw, Poland, \texttt{michal.pilipczuk@mimuw.edu.pl}}
\and
Moritz Venzin\thanks{EPFL, Switzerland, \texttt{moritz.venzin@epfl.ch}}
\and
Robert Weismantel\thanks{ETH Z\"urich, Switzerland, \texttt{robert.weismantel@ifor.math.ethz.ch}}
}

\date{}

\begin{document}
\maketitle

\begin{abstract}
 We consider the problem of solving integer programs of the form $\min \{\,c^\trans x\ \colon\ Ax=b, x\geq 0\}$, where $A$ is a multistage stochastic matrix in the following sense: the primal treedepth of $A$ is bounded by a parameter~$d$, which means that the columns of $A$ can be organized into a rooted forest of depth at most $d$ so that columns not bound by the ancestor/descendant relation in the forest do not have non-zero entries in the same row. We give an algorithm that solves this problem in fixed-parameter time $f(d,\|A\|_{\infty})\cdot n\log^{\Oh(2^d)} n$, where $f$ is a computable function and $n$ is the number of rows of $A$. The algorithm works in the strong model, where the running time only measures unit arithmetic operations on the input numbers and does not depend on their bitlength. This is the first fpt algorithm for multistage stochastic integer programming to achieve almost linear running time in the strong sense.
 For the case of two-stage stochastic integer programs, our algorithm works in time $2^{(2\|A\|_\infty)^{\Oh(r(r+s))}}\cdot n\log^{\Oh(rs)} n$, which gives an improvement over the previous methods both in terms of the polynomial factor and in terms of the dependence on $r$ and~$s$. In fact, for $r=1$ the dependence on $s$ is asymptotically tight assuming the Exponential Time Hypothesis. Our algorithm can be also parallelized: we give an implementation in the PRAM model that achieves running time $f(d,\|A\|_{\infty})\cdot \log^{\Oh(2^d)} n$ using $n$ processors. 
 
The main conceptual ingredient in our algorithms is a new proximity result for multistage stochastic integer programs. We prove that if we consider an integer program $P$, say with a constraint matrix $A$, then for every optimum solution to the linear relaxation of $P$ there exists an optimum (integral) solution to $P$ that lies, in the $\ell_{\infty}$-norm, within distance bounded by a function of $\|A\|_{\infty}$ and the primal treedepth of $A$.
On the way to achieve this result, we prove a generalization and considerable improvement of a structural result of Klein for multistage stochastic integer programs. 
Once the proximity results are established, this allows us to apply a treedepth-based branching strategy guided by an optimum solution to the linear relaxation.  
\end{abstract}


\ERCagreement

\pagebreak
\section{Introduction}\label{sec:intro}

We consider integer optimization programs of the form 
\begin{equation}\label{eq:intro}
\min \{\,c^\trans x\ \colon\ Ax=b, x\geq 0, x\textrm{ is integral}\,\},
\end{equation}
where $x$ is a vector of $m$ variables, $Ax=b$ is a system of $n$ equations, and $c$ is the {\em{optimization goal}} --- an integer vector of dimension $m$. Our goal is to exploit specific structural properties of the {\em{constraint matrix}} $A$ in order to develop efficient solvers for programs posessing these properties. Precisely, we assume that $A$ is either {\em{two-stage}} or {\em{multistage stochastic}}, which are the block forms depicted below.

\begin{figure}[h!]
\centering
\includegraphics[width=0.85\textwidth]{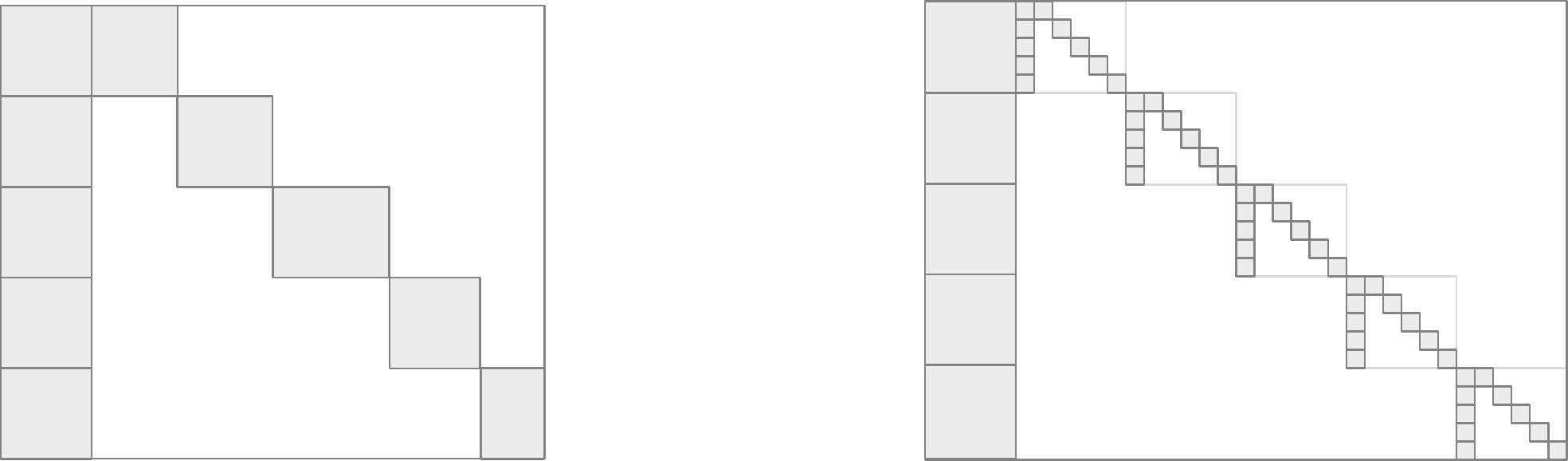}
 \caption{A schematic view of a two-stage matrix (left panel) and a multistage stochastic matrix (right panel). All non-zero entries are contained in the blocks depicted in grey.}
\end{figure}


Formally, we say that $A$ is {\em{$(r,s)$-stochastic}} (left panel above) if after deleting the first $r$ columns, the matrix can be decomposed into blocks with at most $s$ columns each. Thus, two-stage stochastic programs naturally model combinatorial problems with the following structure: there are few (at most~$r$) global variables, so that after setting them the problem can be decomposed into multiple much smaller subproblems --- each involving at most $s$ variables. Due to this, two-stage stochastic integer programming has found multiple applications; see the survey of Schultz et al.~\cite{SchultzSvdV96} for some examples.

Multistage stochastic integer programming is a generalization of the two-stage variant above obtained by allowing further recursive levels in the block structure (right panel above). This idea is probably best captured by the definition of the {\em{primal treedepth}} of a matrix. The primal treedepth of $A$, denoted $\tdP(A)$, is the least integer $d$ such that the columns of $A$ can be organized into a rooted forest of depth at most $d$ (called an {\em{elimination forest}}) with the following property: for every pair of columns that are not independent --- they have non-zero entries in the same row --- these columns have to be in the ancestor/descendant relation in the forest. The form presented in the right panel above can be then obtained by ordering the columns as in the top-down depth-first traversal of the elimination forest, and applying a permutation to the rows to form the blocks.

It turns out that primal treedepth is indeed a structural parameter that can be algorithmically useful in the design of integer programming solvers. By this, we mean the existence of fixed-parameter algorithms for the parameterization by $\tdP(A)$ and $\|A\|_\infty$. For this parameterization, fixed-parameterized tractability can be understood in two ways:
\begin{itemize}[nosep]
 \item {\em{Weak fpt algorithms}} have running times of the form $f(\tdP(A),\|A\|_\infty)\cdot |P|^{\Oh(1)}$, where $f$ is a computable function and $|P|$ is the total bitlength of the encoding of the input program.
 \item From {\em{strong fpt algorithms}} we require time complexity of the form $f(\tdP(A),\|A\|_\infty)\cdot n^{\Oh(1)}$, where $f$ is computable and $n$ is the number of rows of the input matrix. Such algorithms work in the model where input numbers occupy single memory cells on which unit-cost arithmetic operations are allowed. Note that thus, the running time is not allowed to depend on the bitlength of the input numbers.
\end{itemize}
A weak fixed-parameter algorithm for the considered parameterization follows implicitly from the work of Aschenbrenner and Hemmecke~\cite{AschenbrennerH07}. The first to explicitly observe the applicability of primal treedepth to the design of fpt algorithms for integer programming were Ganian and Ordyniak~\cite{GanianO16}, although their algorithm also treats $\|b\|_{\infty}$ as a parameter besides $\tdP(A)$ and $\|A\|_\infty$. A major development was brought by Kouteck\'y et al.~\cite{KouteckyLO18}, who gave the first strong fpt algorithm, with running time $f(\tdP(A),\|A\|_\infty)\cdot n^3\log^2 n$. 
We refer the reader to the joint manuscript of Eisenbrand et al.~\cite{arxiv-IP}, which comprehensively presents the recent developments in the theory of block-structured integer programming. Corollaries~93 and~96 there discuss the cases of two-stage and multistage stochastic integer programming.

The algorithm proposed by Kouteck\'y et al.~\cite{KouteckyLO18} relies on the idea of {\em{iterative augmentation}} using elements of the {\em{Graver basis}} (see also~\cite{EisenbrandHK18,HemmeckeOR13}). Intuitively, the Graver basis of a matrix $A$ comprises of ``single steps'' in the lattice of points $x$ satisfying the system of equations $Ax=0$; see \cref{sec:prelims} for a formal definition. The idea is that if all the elements of the Graver basis are somehow short, then this allows for a local search approach to integer programming. Somewhat simplifying the idea, we start with some candidate for a solution and iteratively improve it by searching for a better one within a ball of small radius. The assumption on the lengths of the Graver basis elements is used to bound the radius of local search necessary to guarantee the correctness. It turns out that in the case of multistage stochastic programs, the $\ell_\infty$-norms of the elements of the Graver basis of the constraint matrix $A$ can be bounded by $g(\tdP(A),\|A\|_\infty)$ for some computable function $g$. This makes the iterative augmentation technique applicable in this setting. 

Let us note that  Kouteck\'y et al.~\cite{KouteckyLO18} relied on bounds on the function $g$ above due to Aschenbrenner and Hemmecke~\cite{AschenbrennerH07}, which only guaranteed computability. However, better and explicit bounds on $g$ were later given by Klein~\cite{Klein20}; see also the joint manuscript of Eisenbrand et al.~\cite{arxiv-IP} for a streamlined presentation. Roughly speaking, the proof of Klein~\cite{Klein20} shows that $g(d,a)$ is at most $d$-fold exponential, and it is open whether this bound can be improved to an elementary function.
Very recently, Jansen et al.~\cite{JansenKL20} showed that the doubly-exponential dependence of the running time on $s$ is in fact necessary for algorithms solving two-stage stochastic integer programs, assuming the Exponential Time Hypothesis.

While robust and elegant, iterative augmentation requires further arguments to accelerate the convergence to an optimal solution in order to guarantee a good running time. As presented in~\cite{arxiv-IP}, to overcome this issue one can either involve the bitlength of the input numbers in measuring the running time, thus effectively resorting to weak fpt algorithms, or reduce this bitlength using technical arguments. For instance, integer program~\eqref{eq:intro} can be solved in time $f(\tdP(A),\|A\|_{\infty})\cdot n^{1+o(1)}\cdot \log^d\|c\|_{\infty}$. However, to the best of our knowledge, there is no known strong fpt algorithm that would achieve a subquadratic running time dependence on $n$, even in the setting of two-stage stochastic integer programming.

In a very recent work, Cslovjecsek et al.~\cite{CslovjecsekEHRW20} gave nearly linear-time strong fpt algorithms for the related setting of {\em{$N$-fold}} and {\em{tree-fold}} integer programming. Here, we respectively consider integer programs~\eqref{eq:intro} in which the {\em{dual}} of the constraint matrix is $(r,s)$-stochastic for some small $r,s$, or has low primal treedepth (we then say that the matrix itself has low {\em{dual treedepth}}). They gave an algorithm that runs in time $f(d,\|A\|_\infty)\cdot m\log^{\Oh(2^d)} m$, where $d$ is the dual treedepth of $A$, $m$ is the number of columns of $A$, and $f$ is a doubly-exponential function. The algorithm can be also parallelized: it can be implemented in the PRAM model so that the running time is $f(d,\|A\|_\infty)\cdot\log^{\Oh(2^d)} m$ on $m$ processors. 

The approach of Cslovjecsek et al.~\cite{CslovjecsekEHRW20} is quite different from iterative augmentation. The key component is a proximity result for integer programs with bounded dual treedepth: they show that if $P$ is an integer program with constraint matrix $A$, then for every optimal solution $x^\frc$ to a suitable\footnote{In~\cite{CslovjecsekEHRW20} this is not the standard relaxation of $P$; see the work for details.} linear relaxation of $P$ there exists an optimal (integral) solution $x^\itg$ to $P$ such that $\|x^\itg-x^\frc\|_1$ is bounded by a function of $\|A\|_\infty$ and the dual treedepth of $A$. It follows that if a solution $x^\frc$ is available, then an optimal integral solution $x^\itg$ can be found in linear fpt time using dynamic programming, where the bound on $\|x^\itg-x^\frc\|$ is used to limit the number of relevant states. This approach requires devising an auxiliary algorithm for solving linear relaxations with bounded dual treedepth in strong fpt time. This is achieved through recursive Laplace dualization using ideas from Norton et al.~\cite{NortonPT92}.

\paragraph*{Our contribution.} In this work, we show that the general proximity framework proposed by Cslovjecsek et al.~\cite{CslovjecsekEHRW20} can be also applied to multistage stochastic integer programming (that is, for programs with bounded primal treedepth). We obtain an algorithm for integer programming~\eqref{eq:intro} that runs in time $f(d,\|A\|_\infty)\cdot n\log^{\Oh(2^{d})} n$ in the strong sense, where $f$ is a computable function, $d=\tdP(A)$, and $n$ is the number of rows of $A$ (\cref{thm:main-multi}). Similarly to~\cite{CslovjecsekEHRW20}, the algorithm can be also parallelized so that, in the PRAM model, it runs in time $f(d,\|A\|_\infty)\cdot \log^{\Oh(2^{d})} n$ on $n$ processors (here, we assume that a suitable elimination forest is provided on input). While the obtained function $f$ is non-elementary in general (it is $d$-fold exponential), we also give a fine-tuned variant of our result that provides more explicit bounds for $(r,s)$-stochastic programs: the running time on $n$ processors is bounded by $2^{(2\|A\|_\infty)^{\Oh(r(r+s))}}\cdot \log^{\Oh(rs)} n$ (\cref{thm:main-two}).

In our proof, we rely on the following two ingredients:
\begin{enumerate}[label=(\alph*),ref=(\alph*),leftmargin=*]
 \item\label{i:proximity} {\bf{Proximity:}} We prove that given a program $P$ in the form~\eqref{eq:intro}, for every optimal solution $x^\frc$ to the linear relaxation of $P$ there exists an optimal (integral) solution $x^\itg$ of $P$ such that $\|x^\itg-x^\frc\|_\infty$ is bounded by a computable function of $\tdP(A)$ and $\|A\|_\infty$. (\cref{lem:proximity-multi})
 \item\label{i:relaxation} {\bf{Relaxation:}} We prove that linear relaxations of programs in the form~\eqref{eq:intro} can be solved in time $\log^{\Oh(2^d)} n$ on~$n$ processors in the PRAM model, where $d\coloneqq\tdP(A)$ and a suitable elimination forest is given. (\cref{lem:relaxation-multi})
\end{enumerate}
Ingredients~\ref{i:proximity} and~\ref{i:relaxation} can be then combined into a standard treedepth-based branching algorithm, which recursively guesses the value of the first (root) variable and splits the problem into independent subproblems of smaller depth. The value of the first variable is guessed from $\{\,\xi\in \Z_{\geq 0}\ |\ |\xi-x^\frc_1|\leq \rho\,\}$, where $x^\frc_1$ is the first entry of an optimal solution to the linear relaxation (computed using~\ref{i:relaxation}) and $\rho$ is the proximity bound provided by~\ref{i:proximity}. The key point is that $\rho$ is bounded by a function of $\tdP(A)$ and $\|A\|_\infty$ only, which allows us to choose the value of the first variable from a set of bounded size.

Ingredient~\ref{i:proximity} presents the main conceptual contribution of this work, and we believe that it uncovers a fundamental property of multistage stochastic integer programs. The proof relies on upper bounds on the $\ell_\infty$-norms of Graver basis elements for multistage stochastic matrices, and in fact our proximity result can be considered a significant generalization of these bounds.
More precisely, the crucial ingredient is a strengthening of a structural result of Klein~\cite{Klein20}, which allows us to  bound the $\ell_\infty$-norm of the projection of Graver-basis elements to the space of stochastic variables (i.e. variables in the first block of columns).

We remark that besides strengthening the formulation, we actually present a new proof of the structural result of Klein, which also yields improved bounds. The impact of the improvement can be illustrated for the case  of integer programs with $(r,s)$-stochastic matrices. The currently most efficient strong fpt algorithm for solving such integer programs, presented in~\cite{arxiv-IP}, is based on the approach of Klein~\cite{Klein20} and runs in time $2^{(2\|A\|_\infty)^{\Oh(r^2 s + rs^2)}} \cdot n^{\Oh(1)}$. On the other hand, our algorithm for this problem runs in time $2^{(2\|A\|_\infty)^{\Oh(r(r+s))}}\cdot n\log^{\Oh(rs)} n$, hence we obtain an improvement not only in the polynomial factor in $n$ but also in the parametric dependence on $r$ and~$s$. Importantly, for $r=1$ this parametric dependence is asymptotically tight: as Jansen et al.~\cite{JansenKL20} have shown, it would contradict the Exponential Time Hypothesis if there existed a $2^{2^{o(s)}}\cdot n^{\Oh(1)}$-time algorithm for solving integer programs whose matrices are $(1,s)$-stochastic and have all coefficients bounded by a constant in absolute values.

Let us stress that our proximity bound~\ref{i:proximity} requires a different proof using completely different tools than the one obtained for tree-fold integer programs by Cslovjecsek et al.~\cite{CslovjecsekEHRW20}. Note also that contrary to~\cite{CslovjecsekEHRW20}, our proximity result concerns the standard linear relaxation.

For ingredient~\ref{i:relaxation}, a direct application of linear programming duality reduces finding optimum values of linear programs with bounded primal treedepth to the case of bounded dual treedepth, which can be treated using the results of Cslovjecsek et al.~\cite{CslovjecsekEHRW20}. However, this only applies to finding the optimum {\em{value}} for a program, and not recovering an example optimal solution, which is what we need for branching. In a nutshell, to overcome this we use complementary slackness conditions to show that an optimal solution to the original problem can be recovered using calls to a solver for the dual. While conceptually natural, this is surprisingly non-trivial on the technical level.

\paragraph*{Organization of the paper.} After giving preliminaries in \cref{sec:prelims}, we outline our algorithms in \cref{sec:algorithm}, deferring the implementation of the necessary ingredients to the later sections. In \cref{sec:stronger-klein-bound} we present the new, stronger proof of the structural result of Klein~\cite{Klein20}, while in \cref{sec:proximity} we use it to establish the proximity results, that is, ingredient~\ref{i:proximity}.  Solving the linear relaxation, that is, ingredient~\ref{i:relaxation}, is discussed in \cref{sec:lp}.


\section{Preliminaries}\label{sec:prelims}

\paragraph*{Model of computation.} We assume a real RAM model of computation, where each memory cell stores a real number (of arbitrary bitlength and precision) and arithmetic operations (including rounding) are assumed to be of unit cost. 

As we will most often work with {\em{sparse matrices}}, that is, matrices where the number of non-zero entries is significantly smaller than the product of the dimensions, we assume that matrices are given as a list consisting of their non-zero entries. Each such entry is given by specifying the row and the column index together with the value of the entry, stored in a single memory cell. Note that we do not take into account the length of the bit encodings of the entries, hence all polynomial factors in our algorithms are in fact strongly polynomial.


\paragraph*{PRAMs.} For the model of parallel computation, we use the CRCW (concurrent read, concurrent write) model of PRAM, where upon concurrent write attempts to the same memory cell, the processor with the smallest index is the one that succeeds.

\paragraph*{(Integer) linear programming.}
We consider (integer) linear programs of the following form:
\begin{gather}
 \min c^\trans x\nonumber \\
 Ax = b \tag{\ilpform}\label{eq:ilp-eq-form}\\
 x\geq 0\nonumber
\end{gather}
Here:
\begin{itemize}[nosep]
 \item $x$ is a vector of $m$ variables;
 \item $A$ is an integer matrix with $m$ columns and $n$ rows; and
 \item $b$ and $c$ are integer vectors of length $m$ and $n$, respectively.
\end{itemize}
The matrix $A$ is called the {\em{constraint matrix}} and $c$ is the {\em{optimization goal}}. Formally, a {\em{linear program}} in the form~\eqref{eq:ilp-eq-form} is a $4$-tuple $P=(x,A,b,c)$ as above. 

For a linear program $P=(x,A,b,c)$ in the form~\eqref{eq:ilp-eq-form}, we denote by $\sol^\R(P)$ and $\sol^\Z(P)$ the sets of fractional and integral solutions to $P$, respectively. That is, $\sol^\R(P)$ the polytope consisting of $x\in \R^n_{\geq 0}$ satisfying $Ax=b$, while $\sol^\Z(P)$ comprises all integer vectors in $\sol^\R(P)$. Further, we define $\opt^\R(P)$ and $\opt^\Z(P)$ as the optimum values of fractional and integral solutions to $P$, respectively. That is, 
$$\opt^\R(P)=\inf\{\,c^\trans x\ \colon\ x\in \sol^\R(P)\,\}\qquad\textrm{and}\qquad\opt^\Z(P)=\inf \{\,c^\trans x\ \colon\ x\in \sol^\Z(P)\,\}.$$
A vector $x\in \sol^\R(P)$ is an {\em{optimal fractional solution}} to $P$ if $c^\trans x=\opt^\R(P)$. \emph{Optimal integral solutions} are defined analogously.

This notation is extended to linear programs in other forms in the obvious way.

\paragraph*{Stochastic matrices.}
A matrix $M$ is {\em{block-decomposable}} if it can be presented as $M=\begin{pmatrix} M' & \\ & M''\end{pmatrix}$, where $M',M''$ are blocks, each containing at least one column or at least one row. Note that we allow corner cases where, for instance, $M''$ has no columns, which means that $M$ has as many all-zero rows at the bottom as $M''$ has rows. The {\em{block decomposition}} of $M$ is the unique presentation of $M$ as $$M=\begin{pmatrix} M_1 & & \\ & M_2 & \\ & & \ddots & \\ & & & M_t\end{pmatrix},$$ where blocks $M_1,\ldots,M_t$ are not block-decomposable.

For nonnegative integers $r$ and $s$, a matrix $A$ is \emph{$(r,s)$-stochastic} if the following condition holds: if $A'$ is $A$ with the first $r$ columns removed, then each block in the block decomposition of $A'$ has at most $s$ columns. Equivalently, an $(r,s)$-stochastic matrix can be written as
\begin{equation}\tag{\stochform}\label{eq:stochastic}
A=\begin{pmatrix} A_1 & B_1 & & & \\ A_2 & & B_2 & & \\ \vdots & & & \ddots & \\ A_t & & & & B_t\end{pmatrix},\end{equation}
where the blocks $A_1,\ldots,A_t$ have $r$ columns and each block $B_i$ has at most $s$ columns. 
As usual, in~\eqref{eq:stochastic} and throughout the paper, empty spaces denote blocks filled with zeros. In general, a presentation of matrix $A$ as in~\eqref{eq:stochastic} shall be called a {\em{stochastic decomposition}} of $A$.

The {\em{primal treedepth}} of a matrix $A$, denoted $\tdP(A)$, can be defined in multiple equivalent ways. For us, it will be convenient to rely on a recursive approach. First, we recursively define the {\em{depth}} of $A$, denoted $\depth(A)$:
\begin{itemize}
\item if $A$ has no columns, then its depth is $0$;
\item if $A$ is block-decomposable, then its depth is equal to the maximum among the depths of the blocks in its block decomposition; and
\item if $A$ has at least one column and is not block-decomposable, then the depth of $A$ is one larger than the depth of the matrix obtained from $A$ by removing its first column.
\end{itemize}
Note that the last condition is equivalent to admitting a stochastic decomposition~\eqref{eq:stochastic} where blocks $A_1,\ldots,A_t$ have one column, while blocks $B_1,\ldots,B_t$ have strictly smaller depth.
We note that, by a straightforward induction, in a matrix of depth $d$ every row contains at most $d$ non-zero entries. 

The {\em{primal treedepth}} of a matrix $A$ is the least $d$ with the following property: the rows and columns of $A$ can be permuted so that after applying the permutation, $A$ has depth $d$.

\medskip

For the sake of future applications, we now observe that the block partition can be efficiently computed in the PRAM model.

\begin{lemma}\label{lem:block-partition}
 Suppose we are given a matrix $M$ with $n$ rows, in which every row contains at most $\Delta$ non-zero entries. Then in the PRAM model, one can using $n$ processors and in time $\Oh(\Delta\log (\Delta n))$ compute the blocks $M_1,\ldots,M_t$ of the block partition of $M$.
\end{lemma}
\begin{proof}
 Recall that $M$ is given as a list $L$ consisting of all non-zero entries, which has length at most $\Delta n$. We think of each entry as a pair consisting of the row index and the column index. First, we sort $L$ in the lexicographic order on those pairs. This takes time $\Oh(\Delta\log (n\Delta))$ using a standard PRAM implementation of mergesort. Then, we partition the list $L$ into sublists $L_1,\ldots,L_s$ by splitting it between every pair of consecutive entries $(i,j)$ and $(i',j')$ satisfying $i<i'$ and $j<j'$. It is easy to see that thus, the lists $L_1,\ldots,L_s$ exactly contain the non-zero entries of all the non-zero blocks of the block partition of~$M$. Also, splitting $L$ into $L_1,\ldots,L_s$ can be done in time $\Oh(\Delta)$, as each of $n$ processors gets a batch of at most $\Delta$ consecutive entries to check for the necessity of a split. Finally, in constant time we can inspect every pair of consecutive lists $L_i,L_{i+1}$ to check whether an all-zero block should be added between them.
\end{proof}

Observe that using \cref{lem:block-partition}, we can check in time $\Oh((r+s)\log ((r+s)n))$ using $n$ processors whether the input matrix is $(r,s)$-stochastic. Similarly, verifying whether the input matrix has depth at most $d$ can be done in time $\Oh(d^2\log (dn))$ (here, \cref{lem:block-partition} needs to be applied recursively on every~block).

\medskip

Note that in the definition of primal treedepth, we allow that the matrix takes the specified form only after applying a permutation of the columns. In many applications such a permutation can be easily inferred from the construction of the program, but we now show that it can be also computed in linear fpt time without any further assumptions. This also applies to the setting of $(r,s)$-stochastic matrices. In both cases, the following definition will be useful: the {\em{primal graph}} of a matrix $A$, denoted $\GP(A)$, is the graph whose vertex set consists of the columns of $A$, where two columns are considered adjacent if they contain non-zero entries in the same row.

\begin{lemma}
 Given a matrix $A$ with $n$ rows and integer $d$, one can in time $2^{\Oh(d^2)}\cdot n$ either conclude that $\tdP(A)>d$, or compute a permutation of the rows and columns of $A$ for which $A$ has depth at most $d$.
\end{lemma}
\begin{proof}
 As observed in previous works (see e.g.~\cite{KouteckyLO18,arxiv-IP}), the primal treedepth of $A$ coincides with the (graph-theoretic) treedepth of its primal graph $\GP(A)$. In~\cite{ReidlRVS14}, Reidl et al. gave an $2^{\Oh(d^2)}\cdot N$-time algorithm to test whether a given $N$-vertex graph has treedepth at most $d$; note that in our case, $N$ is the number of columns of $A$ which can be bounded by $dn$, so the running time is $2^{\Oh(d^2)}\cdot n$. If the algorithm reports a positive outcome, it also returns an {\em{elimination forest}} of depth at most $d$, which witnesses the value of the treedepth. Hence, we can apply the algorithm of Reidl et al. to $\GP(A)$. If the algorithm returns an elimination forest, then from any pre-order traversal of this forest it is easy to obtain a suitable permutation of the rows and columns of $A$ witnessing that $\depth(A)\leq d$.
\end{proof}

\begin{lemma}
 Given a matrix $A$ with $n$ rows and integers $r,s$, one can in time $2^{\Oh(r\log s)}\cdot n$ either find a permutation of the rows and columns of $A$ in which $A$ becomes $(r,s)$-stochastic, or conclude that no such permutation exists.
\end{lemma}
\begin{proof}
 In terms of the primal graph $\GP(A)$, the problem of verifying whether such a permutation exists boils down to the following: verify whether in $\GP(A)$ there exists a set $X$ of $r$ vertices so that every connected component of $\GP(A)-X$ contains at most $s$ vertices. This problem has been studied under the name {\sc{Component Order Connectivity}} by Drange et al.~\cite{DrangeDH16}, who gave an algorithm with running time $2^{\Oh(r\log s)}\cdot N$ on $N$-vertex graphs. Again, in our case $N$ is the number of columns of $A$ which can be assumed to be at most $(r+s)n$, so the running time is $2^{\Oh(r\log s)}\cdot n$. Hence, we may apply the algorithm of Drange et al. to $\GP(A)$. In case of a positive outcome, the algorithm provides a suitable set~$X$, so we can permute the rows and columns of $A$ so that the columns corresponding to $X$ are to the left and the rows and columns corresponding to connected components of $\GP(A)-X$ form blocks.
\end{proof}

Note that the above lemmas yield sequential fpt algorithms. Investigating analogous results in the setting of parallel parameterized algorithms is an interesting question, which however reaches beyond the scope of this work. Therefore, in our algorithms we will assume that the input matrix is already suitably organized, i.e. it has bounded depth or is $(r,s)$-stochastic.

\paragraph*{Graver bases.} The {\em{conformal (partial) order}} is defined as follows: for vectors $x,y\in \R^n$, we write $x\cleq y$ if for each $i\in \{1,\ldots,n\}$ we have $|x_i|\leq |y_i|$ and $x_iy_i\geq 0$, where $x_i$ and $y_i$ are the $i$th entries of $x$ and $y$, respectively. 

A vector or matrix is {\em{integer}} if all its entries are integers. For an integer matrix $A$, by $\ker^\Z (A) $ we denote the set of all integer vectors from $\ker (A)$. The {\em{Graver basis}} of $A$, denoted $\Graver(A)$, consists of all $\cleq$-minimal vectors of $\ker^\Z (A)$. It follows from Dickson's Lemma that $\Graver(A)$ is always finite, however we are interested in more precise bounds on the lengths of vectors in $\Graver(A)$. For this, for $p\in [1,\infty]$ we define the {\em{$\ell_p$ Graver complexity}} of $A$ as
$$g_p(A)\coloneqq \max_{v\in \Graver(A)} \|v\|_p.$$
We will use the following known bounds on the Graver complexity of matrices.

\begin{theorem}[Lemma~2 of~\cite{EisenbrandHK18}]\label{thm:graver-rows}
 For every integer matrix $A$ with $n$ rows,
 $$g_\infty(A) \le (2n \|A\|_\infty+1)^n.$$ 
\end{theorem}
Note that the Graver basis of a matrix $A$ is completely described by $\ker(A)$, so without changing $\Graver(A)$ we may restrict the rows of $A$ to a maximal linearly independent subset. Then the number of columns $m$ is not smaller than the number of rows $n$. Hence we can derive the following bound that is independent of the number of rows in $A$.

\begin{corollary}\label{cor:graver-columns}
 For every integer matrix $A$ with $m$ columns,
$$g_\infty(A) \le (2m \|A\|_\infty+1)^m.$$
\end{corollary}

We will also use the more general bounds for matrices with bounded primal treedepth.

\begin{theorem}[Lemma~26 of~\cite{arxiv-IP}]\label{thm:graver-td}
 There is a computable function $f\colon \N\times \N\to \N$ such that for every integer matrix $A$,
 $$g_{\infty}(A)\leq f(\tdP(A),\|A\|_\infty).$$
\end{theorem}

We note that the proof of \cref{thm:graver-td} given by Eisenbrand et al.~\cite{arxiv-IP} shows that, roughly speaking, $g_{\infty}(A)$ is bounded by a $d$-fold exponential function of $\|A\|_\infty$, where $d$ is the primal treedepth of $A$.

\section{Algorithms}\label{sec:algorithm}

As discussed in \cref{sec:intro}, our algorithms for stochastic integer programming follow from a combination of two ingredients: proximity results for stochastic integer programs, and algorithms for solving their linear relaxations. These ingredients will be proved in the subsequent sections, while in this section we state them formally and argue how the results claimed in \cref{sec:intro} follow.

As for proximity, we show that in stochastic integer programs, for every optimal fractional solution there is always an optimal integral solution that is not far, in terms of the $\ell_\infty$-norm. Precisely, the following results will be proved in \cref{sec:proximity}.

\begin{lemma}\label{lem:proximity-multi}
 There exists a computable function $f\colon \N\times \N\to \N$ with the following property.
 Suppose $P=(x,A,b,c)$ is a linear program in the form~\eqref{eq:ilp-eq-form}. Then for every optimal fractional solution $x^\frc \in\sol^\R(P)$ there exists an optimal integral solution $x^\itg \in \sol^\Z(P)$ satisfying
 $$\|x^\itg-x^\frc\|_\infty\leq f(\depth(A),\|A\|_\infty).$$
\end{lemma}

\begin{lemma}\label{lem:proximity-two}
 Suppose $P=(x,A,b,c)$ is a linear program in the form~\eqref{eq:ilp-eq-form}, where $A$ is $(r,s)$-stochastic for some positive integers $r,s$. Then for every optimal fractional solution $x^\frc\in\sol^\R(P)$ there exists an optimal integral solution $x^\itg\in\sol^\Z(P)$ satisfying
 $$\|x^\itg-x^\frc\|_\infty\leq 2^{\Oh(r(r+s)\|A\|_\infty)^{r(r+s)}}.$$
\end{lemma}

Note that an $(r,s)$-stochastic matrix has depth at most $r+s$, so \cref{lem:proximity-two} could be seen as a special case of \cref{lem:proximity-multi}. However, \cref{lem:proximity-two} provides a better upper bound on the proximity of $(r,s)$-stochastic matrices.

As for solving linear relaxations, in \cref{sec:lp} we will show the following.

\begin{lemma}\label{lem:relaxation-multi}
 Suppose we are given a linear program $P=(x,A,b,c)$ in the form~\eqref{eq:ilp-eq-form}. Let $n$ be the number of rows of~$A$. Then, in the PRAM model, one can, using $n$ processors and in time $\log^{\Oh(2^{\depth(A)})} n$, compute an optimal fractional solution to $P$. 
\end{lemma}

\begin{lemma}\label{lem:relaxation-two}
 Suppose we are given an $(r,s)$-stochastic linear program $P=(x,A,b,c)$ in the form~\eqref{eq:ilp-eq-form}. Let $n$ be the number of rows of~$A$. Then, in the PRAM model, one can, using $n$ processors and in time $2^{\Oh(r^2+rs^2)}\cdot \log^{\Oh(rs)} n$, compute an optimal fractional solution to $P$.
\end{lemma}

Again, \cref{lem:relaxation-two} differs from \cref{lem:relaxation-multi} by considering a more restricted class of matrices (i.e., $(r,s)$-stochastic), but providing better complexity bounds.

We now combine the tools presented above to show the following theorems. 

\begin{theorem}\label{thm:main-multi}
 There is a computable function $f\colon \N\times \N\to \N$ such that the following holds.
 Suppose we are given a linear program $P=(x,A,b,c)$ in the form~\eqref{eq:ilp-eq-form}. Let $n$ be the number of rows of~$A$. Then, in the PRAM model, one can using $n$ processors and in time $f(\depth(A),\|A\|_\infty)\cdot \log^{\Oh(2^{\depth(A)})} n$ compute an optimal integral solution to $P$.  
\end{theorem}
\begin{proof}
 We consider the following recursive algorithm.
 
 As an opening step, using the algorithm of \cref{lem:block-partition}, in time $\Oh(d\log (dn))$ we verify whether the constraint matrix $A$ is block-decomposable. Further proceedings of the algorithm depend on the outcome of this check.
 
 Suppose first that $A$ is block-decomposable, say $A_1,\ldots,A_t$ ($t\geq 2$) are the blocks in the block decomposition of $A$ (note that these blocks are computed by the algorithm of \cref{lem:block-partition}). Then, $P$ can be decomposed into independent integer linear programs $P_1,\ldots,P_t$ with constraint matrices $A_1,\ldots,A_t$ so that an optimal integral solution to $P$ can be obtained by concatenating optimal integral solutions to $P_1,\ldots,P_t$. Therefore,  it suffices to solve programs $P_1,\ldots,P_t$ recursively in parallel, by assigning $n_i$ processors to program $P_i$, where $n_i$ is the number of rows of $A_i$.
 
 Suppose now that $A$ is not block-decomposable. 
 Using the algorithm of \cref{lem:relaxation-multi}, we find an optimal fractional solution $x^\frc\in\sol^\R(P)$ in time $\log^{\Oh(2^{\depth(A)})} n$. By \cref{lem:proximity-multi}, there exists an optimal integral solution $x^\itg\in\sol^\Z(P)$ such that $\|x^\itg-x^\frc\|_\infty\leq \rho$, where $\rho$ is a constant that depends in a computable manner on $\depth(A)$ and $\|A\|_\infty$. In particular, if we denote the first variable of $x$ by $x_1$, then there exists an optimal integral solution $x^\itg\in\sol^\Z(P)$ satisfying $|x^\itg_1-x^\frc_1|\leq \rho$.
 
 Since $A$ is not block-decomposable, we can write $A=(a_1\ A')$, where $a_1$ is the first column of $A$ and $A'$ is a matrix with $\depth(A')<\depth(A)$. For every $\xi\in [x^\frc_1-\rho,x^\frc_1+\rho]\cap \Z_{\geq 0}$, we consider the linear program $P'(\xi)$ defined as
 \begin{gather*}
 \min (c')^\trans x' \\
 A'x' \leq b-\xi\cdot a_1 \\
 x'\geq 0
\end{gather*}
 where $c'$ and $x'$ are $c$ and $x$ with the first entry removed, respectively. From the observation of the previous paragraph it follows that
 $$\opt^\Z(P)=\min \{\,c_1\xi+\opt^\Z(P'(\xi))\ \colon\  \xi\in [x^\frc_1-\rho,x^\frc_1+\rho]\cap \Z_{\geq 0}\,\},$$
 where $c_1$ is the first entry of $c$.
 Therefore, we solve all the programs $P'(\xi)$ for $\xi\in [x^\frc_1-\rho,x^\frc_1+\rho]\cap \Z_{\geq 0}$ recursively one by one, thus obtaining respective optimal solutions $x'(\xi)$, and among the solutions $x=\begin{pmatrix}\xi\\ x'(\xi)\end{pmatrix}$ we may output the one that minimizes~$c^\trans x$.
 
 The correctness of the algorithm follows directly from \cref{lem:proximity-multi}. As for the running time, observe that when treating a linear program with a block-decomposable constraint matrix, we recursively and in parallel solve programs with constraint matrices that are not block-decomposable. On the other hand, when treating a linear program with a non-block-decomposable constraint matrix, we recursively solve at most $2\rho+1$ programs in a sequential manner, each with a strictly smaller depth. Hence, if by $T[n,\Delta,d]$ we denote the (parallel) running time for programs with $n$ rows, depth at most $d$, and all coefficients bounded in absolute value by $\Delta$, then $T[n,\Delta,d]$ satisfies the recursive inequality:
 $$T[n,\Delta,d]\leq \log^{\Oh(2^d)} n + (2\rho+1)\cdot T[n,\Delta,d-1].$$
 This recursion solves to
 $$T[n,\Delta,d]\leq (2\rho+1)^d\cdot \log^{\Oh(2^d)} n.$$
 Since $\rho$ is bounded by a computable function of $\depth(A)$ and $\|A\|_\infty$, the claimed running time bound follows.
\end{proof}

\begin{theorem}\label{thm:main-two}
 Suppose we are given an $(r,s)$-stochastic linear program $P=(x,A,b,c)$ in the form~\eqref{eq:ilp-eq-form}. Let $n$ be the number of rows of~$A$. Then, in the PRAM model, one can, using $n$ processors and in time $2^{(2\|A\|_\infty)^{\Oh(r(r+s))}}\cdot \log^{\Oh(rs)} n$, compute an optimal integral solution to~$P$.
\end{theorem}
\begin{proof}
 We apply the same algorithm as in the proof of \cref{thm:main-multi}, except that we replace the usage of \cref{lem:relaxation-multi} with \cref{lem:relaxation-two}, and we use the proximity bounds provided by \cref{lem:proximity-two} instead of \cref{lem:proximity-multi}. Since an $(r,s)$-stochastic matrix has depth at most $r+s$, the same time complexity analysis shows that the running time is bounded by
 $$T[n,\|A\|_\infty,r,s]\leq (2\rho+1)^{r+s}\cdot 2^{\Oh(r^2+rs^2)}\cdot \log^{\Oh(rs)} n,$$
 where $\rho\leq 2^{\Oh(r(r+s)\|A\|_\infty)^{r(r+s)}}$ is the bound on proximity provided by \cref{lem:proximity-two}. Thus, the running time is bounded by $2^{(2\|A\|_\infty)^{\Oh(r(r+s))}}\cdot \log^{\Oh(rs)} n$, as claimed.
\end{proof}

\newcommand{\Cr}{\mathscr{C}}
\newcommand{\Ir}{\mathscr{I}}

\section{A stronger Klein bound}
\label{sec:stronger-klein-bound}

 

In this section, we recall a structural result of Klein~\cite{Klein20} and prove a stronger variant, which we will need for our proximity bounds in the next section. Formally, we will prove the following theorem.

\begin{theorem}[Stronger Klein bound]\label{thm:klein_lemma_klein}
	Let $T_1, \ldots, T_n \subseteq \Z^d$ be multisets of integer vectors of $\ell_\infty$-norm at most $\Delta$ such that their respective sums are almost the same in the following sense: there is some $b \in \Z^d$ and a positive integer $\epsilon$ such that
	$$ \Big\|\sum_{v \in T_i}v - b\Big\|_{\infty} < \epsilon\qquad\textrm{for all }i\in \{1,\ldots,n\}.$$
	There exists a function $f(d,\Delta) \in 2^{\Oh(d\Delta)^d}$ such that the following holds. Assuming $\|b\|_{\infty} > \epsilon\cdot f(d,\Delta)$, one can find nonempty submultisets $S_i \subseteq T_i$ for all $i\in \{1,\ldots,n\}$, and a vector ${b'}\in \Z^d$ satisfying $\|{b'}\|_{\infty} \leq f(d,\Delta)$, such that
	$$\sum_{v \in S_i}v = {b'} \qquad\textrm{for all } i \in \{1,\dots,n\}.$$
\end{theorem}

Before we proceed to the proof, let us discuss the aspects in which \cref{thm:klein_lemma_klein} strengthens the original formulation of Klein~\cite[Lemma~2]{Klein20}. First, the formulation of Klein required all the vectors to be nonnegative. Second, we allow the sums of respective multisets to differ slightly, which is governed by the slack parameter $\epsilon$. In the original setting of Klein all sums need to be exactly equal, in our setting this corresponds to $\epsilon$ to be equal to $1$. Finally, the argument of Klein yields a bound on the function $f(d,\Delta)$ that is doubly exponential in $d^2$, our proof improves this dependence to doubly exponential in $d$.
The second aspect will be essential in the proof of the proximity bound, while the last is primarily used for improving the parametric factor in the running time of our algorithms.

\bigskip


The remainder of this section is devoted to the proof of \cref{thm:klein_lemma_klein}. 
We need a few definitions at this point. The \emph{cone} spanned by
 vectors $u_1,\dots,u_k \in \Q^d$ is the set of all nonnegative
linear combinations of $u_1,\dots,u_k$. We write
\begin{equation*}
  \label{eq:8}
  \cone(u_1,\dots,u_k) \coloneqq \left\{ \sum_{i=1}^k \lambda_i v_i \colon \lambda_i \ge 0, \, i=1,\dots,k\right\}. 
\end{equation*}
The \emph{integer cone} spanned by these vectors is the set of all
nonnegative \emph{integer} linear combinations of $u_1,\dots,u_k$. We
write
\begin{equation*}
  \label{eq:9}
    \intcone(u_1,\dots,u_k) \coloneqq \left\{ \sum_{i=1}^k \lambda_i v_i \colon \lambda_i \in \Z_{\geq 0}, \, i=1,\dots,k\right\}. 
  \end{equation*}
  We will also use the corresponding matrix notation:
  \begin{equation*}
    \label{eq:10}
    \cone(A) = \{ A x \colon x\in \R_{\geq 0}^k\} \quad \text{ and } \quad \intcone(A) = \{ A x \colon x\in \Z_{\geq 0}^k\},
  \end{equation*}
  for a matrix $A \in \Q^{d\times k}$, respectively. The following lemma is the crucial observation behind our proof of the improved Klein bound.

\begin{lemma}\label{lcm_trick}
  \label{lem:1}
  Let $A_1,\ldots,A_\ell \in \Z^{d \times d}$ be invertible integer matrices such that  $\|A_i\|_\infty \leq\Delta$ for each $i\in \{1,\ldots,\ell\}$ and 
  \begin{equation}
    \label{eq:1}    
\Cr \coloneqq     \bigcap_{i = 1}^\ell \cone(A_i) \neq \emptyset.
   \end{equation}
   Then the following assertions hold:
   \begin{enumerate}[label=(L\arabic*),ref=(L\arabic*),leftmargin=*]
   \item Let $M$ be the least common multiple of the determinants of  matrices $A_i$. Then   $M\leq 3^{(d\Delta)^d}$. \label{i:lcm}
   \item 
   For each integer vector $v \in  \Cr  \cap \Z^d $ one has $M v \in \Ir$, where
\begin{displaymath}
 \Ir =  \bigcap_{i=1}^\ell \intcone(A_i). 
\end{displaymath}
 \label{i:intcone}   
 \item There exist integer vectors  $v_1,\dots,v_t\in \bigcap_{i = 1}^\ell \Ir $ such that  $\|v_j\|_\infty \leq 2^{\Oh(d\Delta)^d}$ for each $j\in \{1,\ldots,t\}$  and 
   \begin{displaymath}
     \cone(v_1,\dots,v_t) = \Cr.
   \end{displaymath}
   \label{i:cone}
  \end{enumerate}
\end{lemma}
\begin{proof}
By the Hadamard bound, $|\det(A_i)|\leq (d\Delta)^d$ for each $i\in \{1,\ldots,\ell\}$. Hence, the least common multiple of the determinants of matrices $A_i$ is bounded by the least common multiple of all the integers in the range $\{1,\ldots,(d\Delta)^d\}$, which in turn is bounded by $3^{(d\Delta)^d}$~\cite[Theorem~1]{lcm_hanson_1972}. This implies \ref{i:lcm}.

We now move to assertion~\ref{i:intcone}. Consider any $i\in \{1,\ldots,\ell\}$. Since $v \in\cone(A_i)\cap \Z^d$, there exists $x_{i} \in \mathbb{Q}^d_{\geq 0}$ with $v = A_i x_{i}$. In fact, one has $x_{i} = {A_i}^{-1} v$. Then Cramer's rule implies that
  \begin{displaymath}
    x_{i} = y_{i} / |\det(A_i)|\qquad \text{ for some } y_{i} \in \Z_{\geq 0}^d. 
  \end{displaymath}
  This shows that if we denote $$M \coloneqq \lcm\left(\{|\det(A_i)|\colon i\in \{1,\ldots,\ell\}\}\right),$$ then 
  \begin{equation*}
    \label{eq:2}    
    M \cdot v \in \intcone(A_i)\qquad\textrm{for each } i\in \{1,\ldots,\ell\}.
  \end{equation*}
  This establishes \ref{i:intcone}.
  
  Finally, by standard arguments~\footnote{This is a simple consequence of the Farkas-Minkowski-Weyl theorem and its proof in \cite[Corollary 7.1a]{schrijver1998theory}, we briefly sketch the argument. The cones $\cone(A_i)$ are finitely generated and as such admit a \emph{inequality description} of the form $\{c_j^\trans x \geq 0\}_{j \in I}$. Note that $c_j \in \Z^d$ and we may assume that $\|c_j\|_{\infty} \leq (d \Delta)^{d/2}$ due to the Hadamard bound. Cone $\Cr$ can then be obtained by conjunction of all the inequalities describing the cones $\cone(A_i)$. It follows that $\Cr$ is polyhedral, and thus is finitely generated. A generator $w$ of $\Cr$ satisfies $d-1$ of the inequalities of the form $c_i^\trans x \geq 0$ with equality, i.e., $w$ is orthogonal to $d-1$ vectors $c_k$. Thus, $w$ is in the kernel of some matrix with $d-1$ rows consisting of vectors $c_k^\trans \in \Z^{1\times d}$. We may take $w$ to be integral and, by the Hadamard bound and using that $\|c_k\|_{\infty} \leq (d\Delta)^{d/2}$, with $\ell_{\infty}$-norm bounded by
  	\begin{eqnarray*}
  		\|w\|_\infty& \leq & \left(d (d\Delta)^{d/2} \right)^{d/2} \leq (d\Delta)^{d^2}. 
  \end{eqnarray*}}, 
there exist  integer vectors $w_1,\dots,w_t\in \Cr\cap \Z^d$ with $\ell_\infty$-norm bounded by $(d\Delta)^{d^2}$ that generate the cone $\Cr$.
	
  By \ref{i:lcm} and \ref{i:intcone}, we see that vectors
  \begin{displaymath}
    v_1\coloneqq M w_1,\qquad v_2\coloneqq M w_2, \qquad \ldots \qquad v_t\coloneqq M w_t 
  \end{displaymath}
  satisfy \ref{i:cone}. 
 \end{proof}

 
 With \cref{lcm_trick} established, we may proceed to the main part of the proof.
 
 \begin{proof}[Proof of \cref{thm:klein_lemma_klein}]
   The proof is along the lines of Klein~\cite{Klein20}. Instead of using the Steinitz lemma, we apply
   Lemma~\ref{lem:1} and we use the Minkowski-Weyl theorem to deal with nonnegative entries.
   In matrix notation, \cref{thm:klein_lemma_klein} can be restated as follows.
   \begin{quote} \it 
   Let $D\coloneqq (2\Delta+1)^d$ and let
   $B\in \Z^{d \times D}$ be a matrix whose columns are all
   possible integer vectors of $\ell_\infty$-norm at most  $\Delta$.  Let
   $m_1,\ldots,m_n \in  \Z_{\geq 0}^{D}$ be nonnegative
   integer vectors such that the vectors $B m_i$ are almost equal, that is, one has 
   \begin{displaymath}
    \|Bm_i - b\|_{\infty} < \epsilon \qquad  \text{for each } \,i \in \{1,\dots,n\},  
  \end{displaymath}
  for some integer vector $b \in \Z^d$ and a positive integer $\epsilon$. Then supposing
  $\|b\|_{\infty} >  \epsilon \cdot 2^{\Oh(d\Delta)^d}$, there
  exist nonzero vectors $m'_1,\ldots,m'_n \in \Z_{\geq0}^{D}$ and ${b'} \in \Z^d$ such that for each $i\in \{1,\ldots,n\}$,
  \begin{enumerate}[label=(K\arabic*),ref=(K\arabic*),leftmargin=*]
  \item\label{i:smaller} $0 \leq m'_i \leq m_i$, 
  \item\label{i:equal}$B{m}_i' = {b'}$, and 
  \item\label{i:bound}
   $\|b'\|_{\infty} \leq 2^{\Oh(d\Delta)^d}$.
  \end{enumerate}
\end{quote}
%
	We focus on proving this formulation. To this end, for each $i \in \{1,\dots,n\}$ we choose some vector $r_i \in \Z_{\geq 0}^D$ such that
	$$B (m_i + r_i) = b.$$
	We can assume that $\|r_i\|_{\infty} \leq \epsilon$, for instance by putting nonzero values in $r_i$ only at entries corresponding to the columns $\{\pm e_1,\ldots,\pm e_d\}$ of $B$, where $\{e_1,\ldots,e_d\}$ is the standard base of $\Z^d$.

        \newcommand{\tu}{\tilde{u}}

	Set $z_i \coloneqq m_i + r_i$. Then the vectors $z_i$ belong to the following (unbounded) polyhedron:
	\begin{alignat*}{2}
		Q \coloneqq \{x \in \R^{D} \mid Bx = b \text{ and } x \geq 0\}.
	\end{alignat*}
	By the Minkowski-Weyl theorem~\cite[Theorem 8.5]{schrijver1998theory}, we
        can write
        \begin{equation}\label{eq:Q}
Q = \conv(\{u_1, \cdots, u_{\ell}\}) + \cone(\{c_1, \cdots,
        c_p\})\qquad \textrm{for some }\,u_1,\ldots,u_\ell, c_1,\ldots,c_p\in \Z^{D}.
        \end{equation} 
        Note that  $u_j\geq 0$ and $Bu_j = b$ for all $j\in \{1,\dots,\ell\}$, and $c_k\geq 0$ and $Bc_k = 0$ for all $k\in \{1,\dots,p\}$. 
        Observe also that $u_j$ are vertex solutions to the linear program defining $Q$. Hence, each $u_j$ has at most $d$ nonzero entries, and there is an invertible submatrix $B_j \in \Z^{d\times d}$ (the basis of $u_j$) consisting of $d$ columns of $B$ such that all nonzero entries of $u_j$ are at the coordinates corresponding to the columns of $B_j$. In particular $\tu_j = B_j^{-1}b$, where $\tu_j$ is $u_j$ projected onto the coordinates corresponding to the columns of~$B_j$. 
        
        Fix $i\in \{1,\ldots,n\}$.
        By~\eqref{eq:Q},
        we can write
        \begin{equation}
          \label{eq:3}
          z_i = \sum_{j=1}^{\ell}\lambda_j u_j + \sum_{k=1}^{p}\mu_k c_k, \qquad \text{ where } \sum_{j=1}^{\ell} \lambda_j = 1 \text{ and }\lambda_j, \mu_k \in \R_{\geq 0}.
        \end{equation}
	From Carathéodory's theorem, see \cite[Corollary 7.1i]{schrijver1998theory}, it follows that we can choose the coefficients $\lambda_1,\ldots,\lambda_\ell$ so that there exists an index $j \in \{1,\dots,\ell\}$ with $\lambda_{j} \geq 1/(d+1)$. Denote this index by $j(i)$. Since all involved vectors  and scalars in~\eqref{eq:3} are nonnegative, we conclude that
        \begin{equation}
          \label{eq:wydra}
          0 \leq u_{j(i)}/(d+1) \leq z_i.
        \end{equation}
        We will now argue that we can find a vector $c\in \Z^d$ and, for each $j\in \{1,\dots,\ell\}$, a nonzero vector ${u}'_j \in \Z_{\geq 0}^D$ such that 
	\begin{alignat}{1}\label{eq:desired_result2}
		B{u}'_{j} = c \qquad \textrm{for all }j \in \{1,\dots,\ell\},
	\end{alignat}
	and
	\begin{alignat}{1}\label{eq:condition_approx}
		{u}'_{j(i)}\leq m_i\qquad \textrm{for all }i \in \{1,\dots,n\}.
	\end{alignat}
	and
	\begin{alignat}{1}\label{eq:final-bound}
		\|c\|_\infty \leq 2^{\Oh(d\Delta)^d}.
	\end{alignat}
	We remark that at the end of the proof we will set ${m}'_i \coloneqq {u}'_{j(i)}$ for all $i \in \{1,\dots,n\}$ and $b' \coloneqq c$. Observe that then, assertions \ref{i:smaller}, \ref{i:equal}, and \ref{i:bound} will immediately follow from \eqref{eq:desired_result2}, \eqref{eq:condition_approx}, and~\eqref{eq:final-bound}, respectively.
        
	Since $Bu_j=B_j\tu_j = b$, it follows that $\Cr\coloneqq \bigcap_{j = 1}^{\ell} \cone(B_j) \neq \emptyset $. By \cref{lem:1}, assertion~\ref{i:cone}, there exist non zero integer vectors 
        $v_1,\dots,v_t \in  \bigcap_{j=1}^\ell \intcone(B_j)$
        of $\ell_\infty$-norm bounded by $2^{\Oh(d\Delta)^d}$
        such that
        \begin{displaymath}
          \cone(v_1,\dots,v_t) = \Cr.  
        \end{displaymath}       
	Since $b/(d+1) \in \Cr$, by Carathéodory's theorem, we can pick at most $d$ vectors of $\{v_1, \dots, v_t\}$, say $v_1, \ldots, v_d$, such that 
	\begin{alignat*}{1}
          b/(d+1) = \sum_{k=1}^d \alpha_k v_k \qquad \textrm{for some }\ \alpha_k\geq 0, \, k\in  \{1,\dots,d\}. 
	\end{alignat*}
	Now we use the assumption on $\|b\|_{\infty}$. Specifically, assume that
        \begin{displaymath}
          \|b\|_{\infty} > (d+1) d \cdot 2\epsilon \cdot \max_i \|v_i\|_\infty = \epsilon \cdot 2^{{\Oh(d\Delta)}^d}. 
        \end{displaymath}
        Observe that there exists an index $k\in \{1,\ldots,d\}$ such that $\alpha_k > 2 \epsilon$. Without loss of generality suppose that $k = 1$. Then we can write
	\begin{alignat*}{1}
		b/(d+1) = 2 \epsilon  v_1 + \underbrace{\left((\alpha_1-2\epsilon)v_1 + \sum_{k=2}^{d}\alpha_k v_k\right)}_{\coloneqq q}.
	\end{alignat*}
	Since $v_1 \in \intcone(B_j)$, for each $j\in \{1,\ldots,\ell\}$ there exists a vector $y_j \in \Z_{\geq 0}^d$ such that $B_jy_j = v_1$. Also, since $q$ is in the cone~$\Cr$,  there exist vectors $x_j \in \R_{\geq 0}^d$ such that $B_jx_j = q$. 
        Thus, for each $j\in \{1,\ldots,\ell\}$ we have
	\begin{alignat*}{1}
		B_j\tu_j/(d+1) = b/(d+1)=2\epsilon v_1+q= B_j(2\epsilon y_j + x_j)
	\end{alignat*}
	Since matrices $B_j$ are invertible, this implies that
	\begin{alignat}{1}\label{eq:bobr}
		\tu_j/(d+1) = 2\epsilon y_j + x_j.
	\end{alignat}
	We set $c=v_1$ and for $j\in \{1,\ldots,\ell\}$, we define ${u}_j'\in \Z_{\geq 0}^D$ as follows:
	\begin{itemize}
	 \item every entry of $u_j'$ that corresponds to a column of $B_j$ is set to the corresponding entry of $y_j$; and
	 \item every other entry of $u_j'$ is set to $0$. 
	\end{itemize}
        Note that thus, vectors $u_j'$ are nonzero,  because $v_1=B_jy_j$ and $v_1$ is nonzero. Also $Bu_j'=B_jy_j=v_1$, so \eqref{eq:desired_result2} also holds. Since we have set $c = v_1$, we have that $\|c\|_{\infty} \leq 2^{\Oh(d\Delta)^d}$, which is~\eqref{eq:final-bound}.
        
        It remains to prove \eqref{eq:condition_approx}. Note that by~\eqref{eq:wydra} and~\eqref{eq:bobr}, for all $i \in \{1,\dots,n\}$ we have
	\begin{alignat*}{1}
		0\leq 2\epsilon y_{j(i)}\leq 2\epsilon y_{j(i)} + x_{j(i)}= \tilde{u}_{j(i)}/(d+1)\qquad\textrm{and}\qquad  u_{j(i)}/(d+1)\leq z_i=m_i+r_i.
	\end{alignat*}
	Since $\|r_i\|_{\infty} \leq \epsilon$ and $y_{j(i)}$ is nonzero, from the above inequalities it follows that $u_{j(i)}' \leq m_i$. This establishes~\eqref{eq:condition_approx} and concludes the proof. 
\end{proof}



\section{Proximity}\label{sec:proximity}
The goal of this section is to prove \cref{lem:proximity-multi} and \cref{lem:proximity-two}. Specifically, we bound the distance between an optimal fractional solution and an optimal integral solution in the case where the constraint matrix has bounded primal treedepth or is $(r,s)$-stochastic.
To facilitate the discussion of proximity, let us introduce the following definition.

\begin{definition}
 Let $P=(x,A,b,c)$ be a linear program in the form~\eqref{eq:ilp-eq-form}. The {\em{proximity}} of~$P$, denoted $\proximity(P)$, is the infimum of reals $\rho\geq 0$ satisfying the following: for every fractional solution $x^\frc\in \sol^\R(P)$ and integral solution $x^\itgo\in\sol^\Z(P)$, there is an integral solution $x^\itg\in \sol^\Z(P)$ such that 
 \[\|x^\itg-x^\frc\|_\infty\leq \rho
 \qquad \text{ and }\qquad x^\itg -x^\frc \cleq x^\itgo - x^\frc.\]
\end{definition}

The condition $x^\itg -x^\frc \cleq x^\itgo - x^\frc$ is equivalent to saying that $x^\itg$ is contained in the axis parallel box spanned by $x^\frc$ and $x^\itgo$, see \cref{i:rectangle}. Intuitively, $x^\itg$ is an integral solution that is close to $x^\frc$ in the $\ell_\infty$-distance while being placed ``in the same direction'' as $x^\itgo$.
\newcommand{\frco}{\bullet}
\begin{figure}[ht]
	\centering
	\begin{tikzpicture}
	\draw[help lines, color=gray!50, dashed] (-0.3,-0.3) grid (6.9,3.4);
	\draw[->,thick] (-0.3,0)--(7,0) node[right]{};
	\draw[->,thick] (0,-0.3)--(0,3.5) node[above]{};
	\node[fill, draw, circle, inner sep=1pt, black, label={0:$x^{\frc}$}] at (5.2,0.5) {};
	\node[label={0:$x \geq 0$}] at (6, 1.3){};
	\node[fill, draw, circle, inner sep=1pt, black, label={180:$x^{\itgo}$}] at (1,3) {};
	\node[fill, draw, circle, inner sep=1pt, black, label={0:$x^{\itg}$}] at (4,2) {};
	\draw[black] (1,3) rectangle (5.2, 0.5);
	\draw[->, dashed, gray] (4,2) -- (1,3);
	\node[fill, draw, circle, inner sep=1pt, black, label={90:$x^{\frco}$}] at (2.2,1.5) {};
	\draw[->, dashed, gray] (5.2, 0.5) -- (2.2,1.5);
	\end{tikzpicture}
	\caption{$x^{\itg} - x^{\frc} \cleq x^\itgo - x^\frc$, $x^\itg$ is in the rectangle spanned by $x^\itgo$ and $x^\frc$.}	
	\label{i:rectangle}
\end{figure}

Previously, the notion of proximity was mostly defined as the maximum distance of any optimal fractional solutions to its closest optimal integral solutions, see for instance~\cite{CookGST86, EisenbrandW19}. Our notion of proximity does not depend on the optimization goal, it is a geometric quantity associated only with the polytope $\sol^\R(P)$.
However, this new notion can also be used to bound the distance of optimal fractional solutions to optimal integral solutions, as the next lemma explains.

\begin{lemma}\label{lem:proximity-optimal}
 Suppose $P=(x,A,b,c)$ is a linear program in the form~\eqref{eq:ilp-eq-form}. Then for every optimal fractional solution $x^\frc$ to $P$ there exists an optimal integral solution $x^\itg$ to $P$ satisfying
 $$\|x^\itg-x^\frc\|_\infty\leq \proximity(P).$$
\end{lemma}
\begin{proof}
 Consider any optimal integral solution $x^\itgo$ to $P$. By the definition of proximity, there is an integral solution $x^\itg$ to $P$ such that $\|x^\itg - x^\frc\|_\infty\le\proximity(P)$ and $x^\itg -x^\frc \cleq x^\itgo - x^\frc$. 
 From the optimality of $x^\itgo$ we get that
 \[c^\trans (x^\itgo-x^\itg)\ge 0.\]
 Let $x^\frco\coloneqq x^\frc + x^\itgo-x^\itg$. This is a fractional solution to $P$, because $Ax^\frco=A(x^\frc + x^\itgo-x^\itg)=b$
 and a straightforward coordinate-wise verification shows that
 \[x^\frc\geq 0,\ x^\itgo\geq 0,\ x^\itg-x^\frc\cleq x^\itgo-x^\frc\quad\text{implies that}\quad x^\frco\geq 0.\]
 The optimality of $x^\frc$ then gives that
 $$c^\trans (x^\itgo -x^\itg) = c^\trans((x^\frc + x^\itgo - x^\itg) - x^\frc)=c^\trans(x^\frco - x^\frc) \le 0.$$
 This implies that $c^\trans x^\itgo = c^\trans x^\itg$, hence $x^\itg$ is also optimal.
\end{proof}

For our main, technical result we will need some additional notation. Suppose that $A$ is a matrix admitting the stochastic decomposition~\eqref{eq:stochastic}. 
Let $x_0,x_1,\ldots,x_t$ be the partition of the vector of variables $x$ so that $x_0$ corresponds to the columns of matrices $A_1,\ldots,A_t$, while $x_i$ corresponds to the columns of $B_i$, for each $i\in \{1,\ldots,t\}$. Partition $c$ into $c_0,c_1,\ldots,c_t$ in the same fashion, and partition $b$ into $b_1,\ldots,b_t$ so that $b_i$ corresponds to the rows of $A_i$ and $B_i$.
In this representation, the program $P$ takes the form:
\begin{align*}
 \min \sum_{i=0}^t c_i^\trans x_i &\\
 A_ix_0+B_ix_i = b_i & \qquad \textrm{for all }i\in \{1,\ldots,t\},\\
 x_i\geq 0 & \qquad \textrm{for all }i\in \{0,1,\ldots,t\}.
\end{align*}
For each $i\in \{1,\ldots,t\}$, let $D_i\coloneqq (A_i\ B_i)$ and 
consider the linear program 
\begin{equation*}
 P_i=\left(\begin{pmatrix}x_0\\ x_i\end{pmatrix},D_i,b_i,0\right),
\end{equation*}
that is, the linear program
\begin{gather*}
 \min 0 \\
 A_ix_0+B_ix_i = b_i,\\
 x_0\geq 0,\quad x_i\geq 0.
\end{gather*}
Note that 
\[\begin{pmatrix}
      x_0 \\ x_1 \\ \vdots\\ x_t
     \end{pmatrix}\in \sol^\R(P)\quad\quad \textrm{if and only if}\quad\quad \begin{pmatrix}x_0\\ x_i\end{pmatrix}\in \sol^\R(P_i)\quad\textrm{for all }i\in \{1,\ldots,t\}.\]
We are now ready to state the main technical result of this section. Intuitively, it provides a single inductive step in the proof of \cref{lem:proximity-multi} and reduces \cref{lem:proximity-two} to the case of matrices with a bounded number of columns.

\begin{theorem}[Composition Theorem]\label{lem:proximity-uberlemma}
 Suppose $P=(x,A,b,c)$ is a linear program in the form~\eqref{eq:ilp-eq-form}, where $A$ admits a stochastic decomposition~\eqref{eq:stochastic}. Adopt the notation presented above and let $k$ be the number of columns of each of the matrices $A_1,\ldots,A_t$. Further, let
 $$
 \gamma\coloneqq \max_{1\leq i\leq t}\ g_\infty(D_i)\qquad \textrm{and}\qquad
 \rho\coloneqq \max_{1\leq i\leq t}
\ \proximity(P_i).$$
 Then
 $$\proximity(P)\leq 3k\gamma\rho\cdot f(k,\gamma)$$
 where $f(k,\gamma)$ is the bound provided by \cref{thm:klein_lemma_klein}.
\end{theorem}
Note that by substituting $f(k,\gamma)$ with the bound provided by \cref{thm:klein_lemma_klein}, we obtain that  $$\proximity(P)\le \rho\cdot 2^{\Oh(k\gamma)^{k}}.$$
Before we prove \cref{lem:proximity-uberlemma}, let us observe the following two consequences of it. As a base case, we give a bound on the proximity of a standard integer program defined by an integer matrix with $m$ columns. This is then first used to bound the proximity of an integer program defined by an $(r,s)$-stochastic matrix. The second consequence is a bound on the proximity of a integer program depending on the primal treedepth of the matrix defining it.

\begin{lemma}\label{lem:proximity-one}
	Let $P=(x,A,b,c)$ be a linear program in the form \eqref{eq:ilp-eq-form} where $A$ has $m$ columns. Then
	$$\proximity(P) \le (m \|A\|_\infty)^{m+1}.$$
\end{lemma}
\begin{proof}
We apply a classical theorem of Cook et al. \cite{CookGST86} to our notion of proximity. Let $x^\frc$ be a fractional solution and $x^\itgo$ an integral solution to \eqref{eq:ilp-eq-form}. Consider the following (integer) linear program:
\begin{gather*}
	\min c^\trans x \\
	Ax = b\\
	x-x^\frc \cleq x^\itgo - x^\frc
\end{gather*}
The constraint $x-x^\frc \cleq x^\itgo - x^\frc$ can be expressed as a conjunction of constraints of the form $x_i^\frc \leq x_i \leq x_i^\itgo$ or $x_i^\itgo \leq x_i \leq x^\frc$ for $i \in \{1, \ldots, m\}$, depending on whether $x_i^\frc \leq x_i^\itgo$ or $x_i^\itgo \leq x_i^\frc$. Thus, the constraint matrix has $m$ columns and its coefficients are bounded by $\|A\|_{\infty}$. By the Hadamard bound it follows that its largest sub-determinant is bounded by $(m\|A\|_{\infty})^m$. By \cite[Theorem 1]{CookGST86} we conclude that there is an integral solution $x^\itg$ such that $\|x^\itg - x^\frc\|_{\infty} \leq m(m\|A\|_{\infty})^m$ and $x^\itg - x^\frc \cleq x^\itgo- x^\frc$.
\end{proof}

\begin{corollary}\label{cor:proximity-two}
	Let $P=(x,A,b,c)$ be a linear program in the form \eqref{eq:ilp-eq-form}, where $A$ is $(r,s)$-stochastic.~Then
	$$\proximity(P)\le
	2^{\Oh(r(r+s)\|A\|_\infty)^{r(r+s)}}.$$
\end{corollary}
\begin{proof}
	By assumption, matrix $A$ admits a decomposition of the form~\eqref{eq:stochastic}, where each block $A_i$ has $r$ columns and each block $B_i$ has at most $s$ columns. 
	Adopting the notation introduced before \cref{lem:proximity-uberlemma}, we see that each matrix $D_i=(A_i\ B_i)$ has at most $r+s$ columns. Applying \cref{lem:proximity-one} to $P_i$, we get 
	\[\proximity(P_i)\le ((r+s) \|A\|_\infty)^{r+s+1}.\]
	Further, by \cref{cor:graver-columns} we have
	\[g_\infty(D_i)\leq (2(r+s)\|A\|_\infty+1)^{r+s}.\] We now combine these two bounds using \cref{lem:proximity-uberlemma} to get the claimed bound on $\proximity(A)$.
\end{proof}

\begin{corollary}\label{cor:proximity-multi}
	There is a computable function $h\colon \N\times \N\to \N$ such that for every linear program $P=(x,A,b,c)$ in the form \eqref{eq:ilp-eq-form}, we have
	$$\proximity(P)\leq h(\tdP(A),\|A\|_\infty).$$
\end{corollary}
\begin{proof}
	We use induction on a more general problem. 
	Suppose $P=(x,A,b,c)$ is a linear program in the form~\eqref{eq:ilp-eq-form}, where $A$ has the following property: removing the first $k$ columns turns $A$ into a matrix of depth at most $\ell$. We would like to prove that
	\[\proximity(A)\le \wh{h}(k,\ell,\|A\|_\infty)\]
	for some computable function $\wh{h}$.
	The corollary then follows by considering the case $k=0$, that is, setting $h(d,\|A\|_\infty)=\wh{h}(0,d,\|A\|_\infty)$.
	
	To prove the general statement we proceed by induction on $\ell$, starting with $\ell=0$. Then $A$ is a matrix with $k$ columns, and, as discussed in \cref{lem:proximity-one}, we can fix a function
	\[\wh{h}(k,0,\|A\|_\infty)\in 
	\Oh(k \|A\|_\infty)^{k+1}.\]
	
	Let us proceed to the induction step for $\ell>0$.
	Since removing the first $k$ columns turns $A$ into a matrix of depth at most $\ell$, it follows that $A$ has a stochastic decomposition~\eqref{eq:stochastic}, where the matrices $A_i$ have $k$ columns each and the matrices $B_i$ have depth at most $\ell$. 
	We may further assume that matrices $B_i$ are not block-decomposable, hence each matrix $B_i$ becomes a matrix of depth at most $\ell-1$ after removing its first column.
	This implies that each matrix $D_i=(A_i\ B_i)$ has the following property: removing the first $k+1$ columns turns it into a matrix of depth at most $\ell-1$. 
	\cref{thm:graver-td} implies that
	\[g_\infty(D_i)\le f(\depth(D_i),\|D_i\|_\infty)\le f(k+\ell,\|A\|_\infty)\]
	for a computable function $f$, while the induction assumption gives
	\[\proximity(P_i)\le \wh{h}(k+1,\ell-1,\|A\|_\infty),\]
	where programs $P_i$ are defined as in the paragraph before \cref{lem:proximity-uberlemma}.
	We may now combine these two bounds using \cref{lem:proximity-uberlemma} to get a bound on $\wh{h}(k,\ell,\|A\|_\infty)$, expressed in terms of $f(k+\ell,\|A\|_\infty)$ and $\wh{h}(k+1,\ell-1,\|A\|_\infty)$.
\end{proof}

Now, \cref{lem:proximity-multi} and \cref{lem:proximity-two} follow by combining \cref{lem:proximity-optimal} with \cref{cor:proximity-multi} and \cref{cor:proximity-two}, respectively.

\subsection{Proof of \cref{lem:proximity-uberlemma}}
\label{sec:proximity-proof}

As we mentioned,
the proof of \cref{lem:proximity-uberlemma} relies heavily on \cref{thm:klein_lemma_klein}: the strengthening of the structural lemma of Klein~\cite{Klein20} that was discussed in \cref{sec:stronger-klein-bound}.

\begin{proof}[Proof of \cref{lem:proximity-uberlemma}]
	Consider any $x^\frc\in \sol^\R(P)$ and $x^\itgo\in \sol^\Z(P)$. Let
	$x^\itg\in \sol^\Z(P)$ be an integral solution such that $x^\itg - x^\frc \cleq x^\itgo - x^\frc$ and subject to the condition that $\|x^\itg - x^\frc\|_1$ is minimized.
	Our goal is to show that then $\|x^\itg - x^\frc\|_\infty \le 3 k \gamma \rho
	\cdot f(k,\gamma)$, where $f(\cdot,\cdot)$ is the function given by \cref{thm:klein_lemma_klein}.
	
	Observe that if there existed a non-zero vector $u\in \ker^\Z(A)$ such that $u\cleq x^\frc -x^\itg$, then we would have that $x^\itg + u \in \sol^\Z(P)$, $(x^\itg+u) - x^\frc \cleq x^\itg - x^\frc \cleq x^\itgo - x^\frc$, and the $\ell_1$ distance from $x^\frc$ to $x^\itg + u$ would be strictly smaller than to $x^\itg$. This would contradict the choice of $x^\itg$.
	Therefore, it is sufficient to show the following: if $\|x^\itg - x^\frc\|_\infty$ is larger than $3 k \gamma \rho\cdot f(k,\gamma)$, then there exists a non-zero vector $u \in \ker^\Z(A)$ such that $u \cleq x^\frc - x^\itg$.
	
	Consider any $i\in \{1,\ldots,t\}$ and
	denote the restrictions of $x^\frc$ and $x^\itg$ to the variables of $P_i$ as follows:
	\[\wt{x}^\frc_i \coloneqq 
	\begin{pmatrix} x^\frc_0\\ x^\frc_i\end{pmatrix}\in \sol^\R(P_i)\qquad
	\text{and}\qquad \wt{x}^\itg_i \coloneqq 
	\begin{pmatrix} x^\itg_0\\ x^\itg_i\end{pmatrix}\in \sol^\Z(P_i).\]
	By the definition of proximity, there is an integral solution
	\[
	\wt{x}_i\in \sol^\Z(P_i)\] such that 
	\[ \|\wt{x}_i - \wt{x}^\frc_i\|_\infty \le \proximity(P_i) \le \rho\qquad \text{ and }\qquad \wt{x}_i - \wt{x}^\frc_i \cleq \wt{x}^\itg_i - \wt{x}^\frc_i.\]
	Since $\wt{x}_i$ and $\wt{x}^\itg_i$ are both integral solutions to $P_i$, we have $\wt{x}_i - \wt{x}^\itg_i \in \ker^\Z(A_i\ B_i)$ and we can decompose this vector into a multiset $G_i$ of Graver elements. That is, $G_i$ is a multiset consisting of sign compatible (i.e., belong to the same orthant) elements of $\Graver(D_i)$ with
	\[\wt{x}_i - \wt{x}^\itg_i=\sum_{g\in G_i}g.\]
	Note that the first $k$ entries of vectors $\wt{x}_1,\dots,\wt{x}_t$ correspond to the same $k$ variables of $P$, but they may differ for different $i\in \{1,\ldots,t\}$.
	
	For a vector $w$, let $\pi(w)$ be the projection onto the first $k$ entries of~$w$. Let $\pi(G_i)$ be the multiset that includes a copy of $\pi(g)$ for each $g\in G_i$.
	By the definition of $\wt{x}^\frc_i$ and $\wt{x}^\itg_i$, we have $\pi(\wt{x}^\frc_i)=\pi(\wt{x}^\frc_j)$ and $\pi(\wt{x}^\itg_i)=\pi(\wt{x}^\itg_j)$ for all $i,j\in \{1,\dots,t\}$.
	From this we get
	\begin{align*}
	\Big\|\sum_{x\in \pi(G_i)}x\ - \pi(\wt{x}^\frc_1-\wt{x}^\itg_1)  \Big\|_\infty
	& = \|\pi(\wt{x}_i) - \pi(\wt{x}^\itg_i) - \pi(\wt{x}^\frc_1) + \pi(\wt{x}^\itg_1)\|_\infty\\
	& = \|\pi(\wt{x}_i) - \pi(\wt{x}^\frc_i)\|_\infty\\
	& = \|\wt{x}_i - \wt{x}^\frc_i\|_\infty\\
	& \le \rho,
	\end{align*}
	for each $i\in \{1,\dots,t\}$.
	Thus, \cref{thm:klein_lemma_klein} is applicable for $d=k$, $\Delta = \gamma$, and $\epsilon = \rho$. Note here that for each $i\in \{1,\ldots,t\}$ and $g\in G_i$, we have $\|g\|_\infty\leq \gamma$.
	In the following we distinguish two cases.

	
	Suppose first that
	\[\|\pi(\wt{x}^\frc_1 - \wt{x}^\itg_1)\|_\infty > \rho\cdot  f(k, \gamma).\]
	By \cref{thm:klein_lemma_klein}, there exist nonempty submultisets $S_1\subseteq \pi(G_1),\dots,S_t\subseteq \pi(G_t)$ such that
	\[\sum_{x\in S_i}x = \sum_{x\in S_j}x\qquad\text{for all } i,j\in \{1,\dots,t\}.\]
	Define a vector $u$ in the following way. For all $i\in\{1,\dots,t\}$, let $\wh{G}_i \subseteq G_i$ be submultisets with $\pi(\wh{G}_i)=S_i$ and set
	\[\wt{u}_i \coloneqq \sum_{g \in \wh{G}_i}g \in \ker^\Z(D_i).\]
	Observe that vectors $\pi(\wt{u}_i)$ are equal for all $i\in\{1,\dots,t\}$. This allows us to define $u$ as the vector obtained by combining all the $\wt{u}_i$, so that projecting $u$ to the variables of $P_i$ yields $\wt{u}_i$, for each $i\in \{1,\ldots,t\}$. Note that since multisets $\wh{G}_i$ are nonempty, $u$ is a non-zero vector. Also $u\in\ker^\Z(A)$, since $\wt{u}_i\in \ker^\Z(D_i)$ for all $i\in \{1,\ldots,t\}$. Further, we have $u \cleq x^\frc -x^\itg$, because for all $i\in \{1,\dots,t\}$,
	\[\wt{u}_i = \sum_{g \in \wh{G}_i}g \cleq \wt{x}_i - \wt{x}^\itg_i \cleq \wt{x}^\frc_i - \wt{x}^\itg_i.\]
	Thus, $u$ satisfies all the requested properties.
	
	We move to the second case: suppose that 
	\[\|\pi(\wt{x}^\frc_1 - \wt{x}^\itg_1)\|_\infty \le \rho\cdot f(k, \gamma).\]
	Since we have $\|\pi(\wt{x}_i - \wt{x}^\itg_i) - \pi(\wt{x}^\frc_1 - \wt{x}^\itg_1)\|_\infty \le \rho$ for all $i\in\{1,\dots,t\}$, we have 
	\[\|\pi(\wt{x}_i - \wt{x}^\itg_i)\|_\infty\leq \rho\cdot f(k,\gamma)+\rho
	\leq 2\rho\cdot f(k,\gamma)\qquad\text{for all }i\in\{1,\dots,t\}.\]
	
	Suppose for a moment that for some $i\in\{1,\ldots,t\}$, there exists an element $g\in G_i$ with $\pi(g)=0$. Then by putting zeros on all the other coordinates, we can extend $g$ to a vector $u\in\ker^\Z(A)$ which satisfies $u\cleq x^\frc - x^\itg$. As $g$ is non-zero, so is $u$, hence $u$ satisfies all the requested properties. Hence, from now on we may assume that no multiset $G_i$ contains an element $g$ with $\pi(g)=0$.
	
	Thus, we have that for all $i\in\{1,\dots,t\}$, the multiset $\pi(G_i)$ consists of non-zero, sign compatible, integral vectors. It follows that
	\[|G_i|=|\pi(G_i)| \le \Big\| \sum_{x\in \pi(G_i)}x\ \Big\|_1
	\le k \Big\| \sum_{x\in \pi(G_i)}x \ \Big\|_\infty  = k \|\pi(\wt{x}_i-\wt{x}^\itg_i)\|_\infty \le 2k\rho\cdot f(k, \gamma).\] 
	Since $\|g\|_\infty \le \gamma$ for every element $g\in G_i$, we infer that
	\[ \|\wt{x}_i - \wt{x}^\itg_i\|_\infty \le \Big\| \sum_{g\in G_i}g \Big\|_\infty
	\le \gamma |G_i| \le 2 k \gamma \rho\cdot f(k, \gamma). \]
	By combining this with $\| \wt{x}_i - \wt{x}^\frc_i \|_\infty \le \rho$, we get
	\[ \|\wt{x}^\itg_i - \wt{x}^\frc_i\|_\infty 
	\le \| \wt{x}^\itg_i - \wt{x}_i \|_\infty + \|\wt{x}_i - \wt{x}^\frc_i \|_\infty 
	\le 2 k \gamma \rho\cdot f(k, \gamma)+\rho \leq 3 k \gamma \rho\cdot f(k, \gamma).\]
	This implies that $\|x^\itg - x^\frc\|\leq 3 k \gamma \rho\cdot f(k, \gamma)$.
\end{proof}

\section{Solving the linear relaxation}\label{sec:lp}

In this section we prove \cref{lem:relaxation-multi} and \cref{lem:relaxation-two}. As mentioned in \cref{sec:intro}, we will rely on results of Cslovjecsek et al.~\cite{CslovjecsekEHRW20}, who considered the dual problem.

\begin{lemma}[Corollary~18 of~\cite{CslovjecsekEHRW20}, with adjusted notation]\label{lem:dual-multi}
 Suppose we are given a linear program $P=(x,A,b,c)$ in the form~\eqref{eq:ilp-eq-form}. Let $n$ be the number of columns of~$A$ and $d\coloneqq \depth(A^\trans)$. Then, in the PRAM model, one can using $n$ processors in time $\log^{\Oh(2^{d})} n$ compute an optimal fractional solution to $P$. 
\end{lemma}

\begin{lemma}[Corollary~17 of~\cite{CslovjecsekEHRW20}, with adjusted notation]\label{lem:dual-two}
  Suppose we are given a linear program $P=(x,A,b,c)$ in the form~\eqref{eq:ilp-eq-form}. Let $n$ be the number of columns of~$A$ and suppose further that $A^\trans$ is $(r,s)$-stochastic. Then, in the PRAM model, one can using $n$ processors in time $2^{\Oh(r^2+rs^2)}\cdot \log^{\Oh(rs)} n$ compute an optimal fractional solution to $P$.
\end{lemma}

In the following, we focus on the proof of \cref{lem:relaxation-multi}. The proof of \cref{lem:relaxation-two} follows from the same line of reasoning, so we only discuss necessary differences at the end. Throughout this section we only work with linear programming without any integrality constraints, so for brevity we drop adjectives ``fractional'' in the notation.

\bigskip

It will be convenient to work with linear programs in the following more general form:
\begin{gather}
 \min c^\trans x\nonumber \\
 Ax \leq b \tag{\ineqform}\label{eq:ilp-ineq-form}\\
 x\geq 0\nonumber
\end{gather}
Note that every linear program in the form~\eqref{eq:ilp-eq-form} can be reduced to a program in the form~\eqref{eq:ilp-ineq-form} by replacing each equality with two inequalities. This reduction preserves the depth of the constraint matrix, as well as being $(r,s)$-stochastic.

Thus, from now on let us fix a linear program $P=(x,A,b,c)$ in the form~\eqref{eq:ilp-ineq-form}; our goal is to compute an optimal solution to $P$.
Let $n$ be the number of rows of $A$ and $d\coloneqq \depth(A)$. We may assume that $A$ has no columns with only zero entries, hence $A$ has at most $dn$ columns.

The reason for using form~\eqref{eq:ilp-ineq-form} is that $P$ admits a simple formulation of the dual linear program. Namely, the dual of $P$ is the following program $P^\trans$:
\begin{gather}
 \max b^\trans y\nonumber \\
 A^\trans y \leq c\nonumber\\
 y\leq 0\nonumber
\end{gather}
We observe that with the help of \cref{lem:dual-multi}, we can efficiently solve $P^\trans$.

\begin{lemma}\label{lem:solve-dual}
 One can compute an optimal solution $y^\frc$ to $P^\trans$ in time $\log^{\Oh(2^d)} n$, using $n$ processors.
\end{lemma}
\begin{proof}
By negating the variables and introducing a vector of slack variables $z$, one for every constraint in~$P^\trans$, solving $P^\trans$ is equivalent to solving the following program $\wo{P}^\trans$:
 \begin{gather}
 \min b^\trans y\nonumber \\
 -A^\trans y + I z= c \nonumber\\
 y\geq 0, z\geq 0\nonumber
\end{gather}
More precisely, optimal solutions of $P^\trans$ can be obtained from optimal solutions of $\wo{P}^\trans$ by dropping variables of $z$ and negating the variables of $y$.
Note that 
$$\depth\left(\begin{pmatrix} -A^\trans & I \end{pmatrix}^\trans\right)=\depth(A)=d.$$
Since $\wo{P}^\trans$ is in the form~\eqref{eq:ilp-eq-form} and its constraint matrix has at most $n+dn$ columns, we may use \cref{lem:dual-multi} to find an optimal solution to $\wo{P}^\trans$ in time $\log^{2^{\Oh(d)}} n$.  Consequently, within the same asymptotic running time we can find an optimal solution $y^\frc$ to $P^\trans$.
\end{proof}

Thus, by applying the algorithm of \cref{lem:solve-dual}, we may assume that we have an optimal solution $y^\frc$ to the dual program $P^\trans$. 
Classic linear programming duality tells us that the optimum values of the programs $P$ and $P^\trans$ are equal. In other words, if we denote 
$$\lambda\coloneqq b^\trans y^\frc,$$
then $\lambda=\opt^\R(P)=\opt^\R(P^\trans)$. Note here that $\lambda$ can be computed from $y^\frc$ in time $\Oh(\log n)$. However, we are interested in computing not only the optimum value of a solution to $P$ --- which is $\lambda$ --- but we would like to actually find some optimal solution.

To this end, we will exploit the knowledge of $y^\frc$ trough the complementary slackness conditions. Let us denote the consecutive variables of $x$ as $x_1,\ldots,x_m$, where $m$ is the number of columns of $A$, and similarly enumerate the entries of $y$, $b$, and $c$. Also, let the consecutive columns of $A$ be $a_{\circ,1},\ldots,a_{\circ,m}$ and the consecutive rows of $A$ be $a_{1,\circ}^\trans,\ldots,a_{n,\circ}^\trans$.
The following claim captures the assertions that can be inferred from the complementary slackness conditions.

\begin{claim}\label{cl:slackness}
 There exists an optimal solution $x^\frc$ to $P$ satisfying the following properties:
 \begin{enumerate}[label=(S\arabic*),ref=(S\arabic*),leftmargin=*]
  \item\label{s:nzero-tight} For every $i\in \{1,\ldots,n\}$ such that $y^\frc_i<0$, we have $a_{i,\circ}^\trans x^\frc = b_i$.
  \item\label{s:ntight-zero} For every $j\in \{1,\ldots,m\}$ such that $a_{\circ,j}^\trans y^\frc < c_j$, we have $x^\frc_j=0$.
 \end{enumerate}
\end{claim}

Let 
$$X\coloneqq \{i\colon y^\frc_i<0\}\subseteq \{1,\ldots,n\}\qquad \textrm{and}\qquad Y\coloneqq \{j\colon a_j^\trans y^\frc<c_j\}\subseteq \{1,\ldots,m\}$$
be the sets of indices to which the implications of \cref{cl:slackness} apply. Before we continue, let us discuss computing $X$ and $Y$ in the PRAM model using $n$ processors.
Obviously, $X$ can be computed in time $\Oh(1)$. As for $Y$, we claim that it can be computed in time $d^{\Oh(1)}\cdot \log n$. Observe that computing the inner products $a_{
\circ,j}^\trans y^\frc$ for all $j\in \{1,\ldots,m\}$ boils down to computing $m$ sums, where the $j$th sum ranges over the list of non-zero entries of the column $a_{
\circ,j}$. Such lists can be computed in time $d^{\Oh(1)}\cdot \log n$ by sorting the list of non-zero entries of $A$ in the lexicographic order (first by the column index and then by the row index), and then splitting it appropriately. Since the total length of the lists is at most $dn$, their sums can be computed in time $d^{\Oh(1)}\cdot \log n$ on $n$ processors by assigning to each list a number of processors proportional to its length.

Thus, we may assume that we have computed the sets $X$ and $Y$. Let $\wo{X}=\{1,\ldots,n\}\setminus X$. We introduce the following notation:
\begin{itemize}[nosep]
 \item Let $\tilde{x}$, $\tilde{c}$, $\tilde{y}$, $\tilde{b}$ be obtained from $x$, $c$, $y$, $b$ by removing all the entries with indices in $Y$, $Y$, $\wo{X}$, and~$\wo{X}$, respectively.
 \item Let $\dbtilde{x}$ be the vector of the remaining variables of $x$, i.e., those with indices in $Y$.
 \item Let $\wt{A}$ be the matrix obtained from $A$ by removing all rows with indices in $\wo{X}$ and all columns with indices in $Y$. 
\end{itemize} 
The notation is extended to solutions $x^\frc$, $y^\frc$, etc. naturally.
Note that all these objects can be computed in time $d^{\Oh(1)}$ using $n$ processors.

Observe that 
\begin{equation}\label{eq:beaver}
 \wt{A}^\trans \tilde{y}^\frc = \tilde{c},
\end{equation}
as in the solution $y^\frc$ to $P^\trans$, all the constraints with indices outside of $Y$ are tight, by the definition of $Y$, while the entries of $y^\frc$ outside of $\tilde{y}^\frc$ are zeros anyway. Further, \cref{cl:slackness} can be rewritten as follows.

\begin{claim}\label{cl:slackness-2}
 There exists an optimal solution $x^\frc$ to $P$ such that
 \begin{equation}\label{eq:slackness-2}
 \wt{A}\tilde{x}^\frc =\tilde{b}\qquad\textrm{and}\qquad \dbtilde{x}^\frc=0.
 \end{equation}
\end{claim}

The next lemma explains the main gain provided by the complementary slackness conditions: if a solution to $P$ satisfies the equalities given in \cref{cl:slackness-2}, then it automatically is an optimal solution.

\begin{lemma}\label{lem:slack-optimal}
 Suppose a vector $x^\frc\in \R_{\geq 0}^m$ satisfies~\eqref{eq:slackness-2}. Then $c^\trans x^\frc=\lambda$.
\end{lemma}
\begin{proof}
 We observe that
 \begin{align*}
  c^\trans x^\frc & = \tilde{c}^\trans \tilde{x}^\frc & & \text{[by } \dbtilde{x}^\frc=0\text{]}\\
                  & = \left(\wt{A}^\trans \tilde{y}^\star\right)^\trans \tilde{x}^\frc & & \text{[by~\eqref{eq:beaver}]} \\
                  & = (\tilde{y}^\star)^\trans \wt{A}\tilde{x}^\frc & & \\
                  & = (\tilde{y}^\star)^\trans \tilde{b} & & \text{[by } \wt{A}\tilde{x}^\frc =\tilde{b}\text{]}\\
                  & = \tilde{b}^\trans \tilde{y}^\star = \lambda, & & \text{[as the value is a scalar]}
 \end{align*}
 as claimed.
\end{proof}
Now, consider the following linear program $\wh{P}$ and recall that $x_1$ denotes the first variable of $x$:
\begin{align*}
 \min x_1 &  \\
 Ax \leq b, & \qquad \wt{A}\tilde{x} = \tilde{b},  \\
 x\geq 0, & \qquad \dbtilde{x}=0.
\end{align*}
By \cref{cl:slackness-2}, there exists an optimal solution to $P$ that is also a feasible solution to $\wh{P}$. On the other hand, by \cref{lem:slack-optimal}, every feasible solution to $\wh{P}$ is actually an optimal solution to $P$. Therefore, there exists an optimal solution $x^\frc$ to $P$ that satisfies the following: $x^\frc_1=\opt^\R(\wh{P})$. 

We now observe that the value $\opt^\R(\wh{P})$ can be computed in time $\log^{\Oh(2^d)} n$ using the same approach as that used in \cref{lem:solve-dual}. Namely, we may eliminate all the variables of $\dbtilde{x}$ from $\wh{P}$ by just substituting them with zeroes, and we can replace each equality from $\wt{A}\tilde{x}=\tilde{b}$ with two inequalities. In this way, we obtain an equivalent linear program in the form~\eqref{eq:ilp-ineq-form} with at most $3n$ constraints, whose constraint matrix has depth at most $d$. Now, using the approach from \cref{lem:solve-dual} we may find an optimal solution to the dual of this program, whose value coincides with $\opt^\R(\wh{P})$.

To summarize, we have argued the following claim.

\begin{claim}\label{cl:get-x1}
In time $\log^{\Oh(2^d)} n$ we may compute a value $\xi$ (equal to $\opt^\R(\wh{P})$) with the following property: there exists an optimal solution $x^\frc$ to $P$ such that $x_1^\frc=\xi$. 
\end{claim}

We now use \cref{cl:get-x1} in the following recursive algorithm for finding an optimal solution to~$P$:
\begin{itemize}
 \item If the constraint matrix $A$ is block-decomposable, say $A_1,\ldots,A_t$ ($t\geq 2$) are the blocks of the block decomposition of $A$, then decompose $P$ into $t$ independent programs $P_1,\ldots,P_t$ with constraint matrices $A_1,\ldots,A_t$, respectively. Solve these programs recursively in parallel, by assigning to each program $P_i$ the number of processors equal to the number of rows of $A_i$. Then combine the obtained optimal solutions to $P_1,\ldots,P_t$ into an optimal solution to $P$.
 \item If the constraint matrix $A$ is not block-decomposable, then it can be written as $(a_{\circ,1}\ A')$, where $a_{\circ,1}$ is the first column of $A$ and $A'$ is a matrix such that $\depth(A')<\depth(A)$. Using \cref{cl:get-x1}, in time $\log^{\Oh(2^d)} n$ find a value $\xi$ such that there exists an optimal solution to $P$ that sets the first variable to $\xi$. Now, consider the linear program $P'$ defined as
\begin{gather*}
 \min (c')^\trans x' \\
 A'x' \leq b-\xi\cdot a_{\circ,1} \\
 x'\geq 0
\end{gather*}
 where $c'$ and $x'$ are $c$ and $x$ with the first entry removed, respectively. Apply the algorithm recursively to $P'$, noting that its constraint matrix $A'$ has a strictly smaller depth than $A$. Finally, an optimal solution to $P$ can be obtained from the computed optimal solution to $P'$ by assigning value $\xi$ to the first variable.  
\end{itemize}

The correctness of the algorithm follows from \cref{cl:get-x1} in a straightforward manner. As for the running time, observe that when the algorithm considers a linear program with a constraint matrix that is not block-decomposable, it recurses on a linear program with a strictly smaller depth. On the other hand, when the algorithm considers a linear program with a block-decomposable constraint matrix, it recurses on several linear programs whose contraint matrices are not block-decomposable. It follows that if the initial linear program $P$ has depth $d$, then the recursion has depth at most $2d$. As each level of the recursion is done in parallel in time $\log^{\Oh(2^d)} n$, the total running time of $\log^{\Oh(2^d)} n$ follows.

This concludes the proof of \cref{lem:relaxation-multi}. \cref{lem:relaxation-two} can be proved in exactly the same manner, except that we replace the usage of \cref{lem:dual-multi} with \cref{lem:dual-two}, noting that the linear programs in question are $(r,s)$-stochastic. Also, the recursion has depth $2(r+s)$ instead of $2d$.

\bibliography{ref}
\clearpage
\appendix
\newcommand{\inappendix}{yes!}
\end{document}